\documentclass{amsart}[12pt]
\usepackage{algorithm}
\usepackage{algpseudocode}
\usepackage{longtable}
\usepackage{dashbox} % For dashed boxes
\usepackage{url}
\usepackage[colorlinks=true,citecolor=blue,linkcolor=blue,urlcolor=blue,bookmarks,bookmarksopen,bookmarksdepth=2,backref=page,breaklinks]{hyperref}

\renewcommand*{\backref}[1]{}
\renewcommand*{\backrefalt}[4]{%
	\ifcase #1 Not cited.%
	\or        Cited on page~#2.%
	\else      Cited on pages~#2.%
	\fi}
\usepackage{tikz-cd}
\usepackage{bookmark}
\bookmarksetup{
	numbered, 
	open,
}

\usepackage{float}
\usepackage{subcaption}

\usepackage{diagbox}
\usepackage{relsize} % pour ajuster la taille

\usepackage{amsmath,amscd}
\usepackage[utf8]{inputenc}
\usepackage{multicol}
\usepackage{tabularx}
\usepackage{xcolor}
\usepackage{colortbl}
\usepackage{enumitem} 
\usepackage{verbatimbox}
\usepackage{listings,chngcntr}
\usepackage{fancyvrb}
\usepackage{graphicx}

\lstdefinestyle{mystyle}{
	basicstyle=\ttfamily,
	columns=flexible,
	showstringspaces=false
}

\usepackage{amsthm}

\theoremstyle{plain}
\newtheorem{theorem}{Theorem}[section]
\newtheorem{lemma}[theorem]{Lemma}

\theoremstyle{definition}

\newtheorem{example}[theorem]{Example}

\newtheorem{definition}[theorem]{Definition}
\newtheorem{remark}[theorem]{Remark}

\numberwithin{theorem}{section}
\numberwithin{equation}{section}
\numberwithin{table}{section}
\numberwithin{figure}{section}
\numberwithin{algorithm}{section}
\subjclass[2020]{11T71, 06E30, 94A60, 05B10}
\keywords{Bent Function, Normality, Difference Set, Boolean Function, Classification}

\makeatletter
\def\blfootnote{\gdef\@thefnmark{}\@footnotetext}
\makeatother

\def\F{\mathbb F}

\def\B{\mathcal B}

\def\rmc(#1,#2){RM(#1,#2)}
\def\form(#1,#2){H(#1,#2)}
\def\val(#1,#2){V(#1,#2)}
\def\boole(#1){ B(#1) }
\def\bst(#1,#2,#3){\boole(#1,#2,#3)}
\def\tildeboole(#1){ \widetilde B(#1) }

\def\fd{{\mathbb F}_2}
\def\fdm{{\mathbb F}^m_2}

\def\aglm{{\textsc{agl}}(m,2)}

\def\der(#1,#2){{\rm Der}_{#2} (#1)}

\def\stab(#1){\textsc{stab}(#1)}
\def\stablevel(#1,#2){\textsc{stab}_{#1}(#2)}

\def\wt{{\rm wt}\,}

\def\rmq(#1,#2,#3){\rmc(#1,#2)/\rmc(#3,#2)}

\def\walsh(#1,#2){\widehat{#1}(#2)}

\def\level(#1){\underset{#1}{=}}

\def\binogauss(#1,#2){\genfrac{[}{]}{0pt}{}{#1}{#2}}

\def\val(#1){{\rm val}(#1)}

\def\anf(#1){{\rm anf}(#1)}
\def\fix(#1,#2,#3,#4){{{\rm fix}^{#1,#2}_{#3}}{(#4)}}
\def\classe(#1,#2,#3){{{\mathcal T}(#1,#2,#3)}}

\def\stab(#1){\textsc{stab}(#1)}
\def\stablevel(#1,#2){\textsc{stab}^{#1}(#2)}
\def\stableveldim(#1,#2,#3){\textsc{stab}^{#1}_{#2}(#3)}
\def\stableveldeg(#1,#2,#3,#4){\textsc{stab}^{#1,#2}_{#3}(#4)}
\def\level(#1){\underset{#1}{\sim}}
\def\bound(#1,#2){\underset{#1}{\overset{#2}{\sim}}}
\def\modulo(#1,#2){\mod\rmc(#1,#2)}

\def\pow#1.#2{\tiny$10^{#1.#2}$}

\def\level(#1){\underset{#1}{=}}
\def\val(#1){{\rm val}(#1)}

\def\anf(#1){{\rm anf}(#1)}
\def\fix(#1,#2,#3,#4){{{\rm fix}^{#1,#2}_{#3}}{(#4)}}
\def\classe(#1,#2,#3){{{\mathcal T}(#1,#2,#3)}}

\def\stab(#1){\textsc{stab}(#1)}
\def\stablevel(#1,#2){\textsc{stab}^{#1}(#2)}
\def\stableveldim(#1,#2,#3){\textsc{stab}^{#1}_{#2}(#3)}
\def\stableveldeg(#1,#2,#3,#4){\textsc{stab}^{#1,#2}_{#3}(#4)}
\def\level(#1){\underset{#1}{\sim}}
\def\bound(#1,#2){\underset{#1}{\overset{#2}{\sim}}}
\def\modulo(#1,#2){\mod\rmc(#1,#2)}

\def\pow#1.#2{\tiny$10^{#1.#2}$}

\def\degV(#1,#2,#3){\deg_{#1+#2}(#3)}
\def\PLUS{$0$} 
\def\MOINS{$1$}

\title{Normality of 8-bit bent functions}

\author{Valérie Gillot}
\author{Philippe Langevin}
\address{Imath, university of Toulon}
\email{valerie.gillot@univ-tln.fr} \email{philippe.langevin@univ-tln.fr}
\author{Alexandr Polujan}
\address{Otto-von-Guericke-Universit\"{a}t}
\email{alexandr.polujan@gmail.com}

\date{\today}

\begin{document}
	
	\maketitle
	\begin{abstract}
		Bent functions are Boolean functions in an even number of variables that are indicators of Hadamard difference sets in elementary abelian 2-groups. A bent function $f$ in $m$ variables is said to be normal if it is constant on an affine space of dimension $m/2$. In this paper, we demonstrate that all bent functions in $m = 8$ variables~--- whose exact count, determined by Langevin and Leander (Des. Codes Cryptogr. 59(1–3): 193–205, 2011), is approximately $2^{106}$~--- share a common algebraic property: every 8-variable bent function is normal, up to the addition of a linear function. With this result, we complete the analysis of the normality of bent functions for the last unresolved case, $m=8$. It is already known that all bent functions in $m$ variables are normal for $m \leq 6$, while for $m \geq 10$, there exist bent functions that cannot be made normal by adding linear functions. Consequently, we provide a complete solution to an open problem by Charpin (J. Complex. 20(2-3): 245-265, 2004).
	\end{abstract}
	\section{Introduction}
	\blfootnote{Parts of this work were presented in~\cite{GLP_BFA_2024} at ‘‘The 9th International Workshop on Boolean Functions and their Applications, BFA 2024’’.}
	
	Let $G$ and $H$ be finite groups written additively. A function $L \colon G \to H$ is called a \emph{homomorphism} if it satisfies the condition $ L(x + a) - L(x) = L(a)$ for all $x, a \in G$. The homomorphisms $L \colon G \to H$ can be viewed as \emph{linear mappings} between the finite groups $G$ and $H$, and they are equivalently characterized by the property:
	$$
	|\{x \in G\colon  L(x+a)-L(x)=b\}| \in\{0, |G|\}.
	$$
	
	Based on this characterization, it is natural to define mappings $F \colon G \to H$ that are considered the ``opposite of linear'' and are thus referred to as \emph{perfect nonlinear}, as follows (see also~\cite{Pott04} and \cite{KoelschPolujan2024}):
	$$|\{x \in G\colon  F(x+a)-F(x)=b\}|=\frac{|G|}{|H|} \quad \mbox{holds for all } a \in G \setminus\{0\}\mbox{ and } b \in H.$$
	
	Perfect nonlinear functions, thanks to their optimal properties, hold significant theoretical and practical importance in various areas of mathematics~\cite{Pott16}. For instance, they serve as valuable primitives in cryptography, offering optimal resistance to linear and differential cryptanalysis when employed in DES-like cryptosystems~\cite{NybergKnudsen92}. Among these functions, two extremal cases have attracted considerable research interest. In the setting $G = H = \mathbb{F}_p^m$, $p$ odd, such functions are referred to as \emph{planar functions}. These have been studied extensively since 1968 in finite geometry, as they enable the construction of finite projective planes~\cite{DembowskiOstrom1968}. On the other hand, when $G = \mathbb{F}_2^m$ and $H = \mathbb{F}_2$, the functions are known as \emph{bent functions}. First introduced by Rothaus~\cite{ROTHAUS1976} in 1976, bent functions have attracted attention from combinatorialists because they provide the construction of difference sets~\cite{MCFARLAND1973,Dillon1974}.
	
	In this article, we focus solely on perfect nonlinear functions $f\colon\F_2^m\to\F_2$, that is, (Boolean) bent functions, which can be viewed as multivariate polynomials in the ring $\F_2[x_1,\ldots,x_m]/(x_1^2+x_1,\ldots,x_m^2+x_m)$. However, a general characterization
	of bent functions remains elusive. What we know is that the number of variables $m$ of a bent function on $\F_2^m$ must be even and its degree (as polynomial) satisfies $2\le\deg(f)\le m/2$; see~\cite{ROTHAUS1976}. Among bent functions of all possible degrees, only quadratic bent polynomials are understood thanks to the 1-to-1 correspondence of quadratic homogeneous bent functions on $\F_2^m$ with invertible symplectic $m \times m$-matrices over $\F_2$ (see~\cite[Chapter 15]{MacWilliamsSloane}) and the fact that addition of an affine function does not affect the bent property (the latter can be seen from the definition of perfect nonlinearity). This connection, in turn, implies the enumeration of all quadratic bent functions in $m=2k$ variables, and their classification: up to the action of the affine general group (which preserves bentness), there exists exactly one quadratic bent function, the  ``scalar product'' $(x,y)\in\F_2^k\times\F_2^k\mapsto \langle x,y \rangle =x_1y_1+\cdots+x_ky_k$.
	
	Although bent functions on $\mathbb{F}_2^k \times \mathbb{F}_2^k$ of all possible degrees can be constructed using the \emph{Maiorana-McFarland construction}~\cite{MCFARLAND1973}, defined by $(x, y) \in \mathbb{F}_2^k \times \mathbb{F}_2^k \mapsto \langle x, \pi(y) \rangle + g(y)$, where $\pi$ is a permutation of $\mathbb{F}_2^k$ and $g$ is an arbitrary Boolean function on $\mathbb{F}_2^k$, the general case beyond quadratic functions is significantly more complex. This complexity is evidenced by progress in solving enumeration and classification problems. Using computational techniques, the exact numbers of bent functions in dimensions 6 and  8, which are $\approx2^{25}$ and $\approx2^{106}$, were obtained in~\cite[p. 258]{Preneel1993} and \cite{PLGL-2011}, respectively. In general, finding good bounds on the number of bent functions is an active research direction, for the recent progress on this subject, we refer to~\cite{PotapovTT24}. Beyond enumeration, computational classification has primarily addressed cubic bent functions in dimensions 6 and 8~\cite{ROTHAUS1976,Dillon1972},~\cite[p. 103]{Braeken2006}, and some specific subclasses of bent functions~\cite{Langevin2011CountingPS,LP_BFA_2024}. Notably, all bent functions in dimension 6 are equivalent to those of the Maiorana-McFarland class~\cite{ROTHAUS1976}. However, finding a unifying construction framework for higher dimensions remains elusive, as computational results~\cite{LP_BFA_2024} show that the known constructions represent only a small fraction of all bent functions in dimension 8.
	
	These computational findings highlight the importance of not only developing new constructions for Boolean bent functions~--- see~\cite{Mesnager2016,Carlet2021} for comprehensive references~--- but also uncovering unifying properties that can characterize ``large'' infinite families of bent functions, as well as describe most, if not all, examples in smaller dimensions.
	
	An example of such a property is normality, introduced by Dobbertin~\cite{DobbertinFSE1994}, which also serves as a measure of the cryptographic complexity of bent functions~\cite[Section 3.1.8]{Carlet2021}. A bent function $f$ in $m=2k$ variables is called \emph{normal} if there exists an affine subspace of dimension $k$ on which $f$ is constant. A broader concept, known as \emph{weak normality}~\cite{BENT}, allows the $2k$-variable bent function to be affine on a flat of dimension $k$. This generalization is meaningful because the bent property is preserved under the addition of affine terms. While most known constructions are weakly normal or normal, identifying bent functions that are non-normal or non-weakly-normal remains a challenging problem both theoretically and computationally. Below, we briefly summarize some important results on the normality of bent functions.
	
	For $m \leq 6$, all bent functions in $m$ variables are known to be normal, thanks to their complete classification. In dimensions $2$ and $4$, bent functions are quadratic due to the degree bound, and hence normal. In dimension 6, there are exactly four equivalence classes of bent functions. Possible sets of representatives for these classes can be found in~\cite{ROTHAUS1976} and~\cite{Carlet1994}; their normality can be verified directly, as demonstrated in~\cite{agievich2008bent}.  
	
	For $m \geq 10$, the existence of non-normal (and non-weakly-normal) bent functions has been proven through a combination of theoretical and computational approaches. The theoretical foundation relies on the result~\cite{BENT} that the direct sum $(x,y)\in\F_2^n\times\F_2^k\mapsto f(x)+g(y)$ of a non-normal (resp. non-weakly-normal) bent function $f$ on $\F_2^n$ and a quadratic, and hence normal, bent function $g$ on $\F_2^k$ is always a non-normal (resp. non-weakly-normal) bent function. Consequently, it suffices to find examples of non-normal and non-weakly-normal bent functions, as most known theoretical constructions are normal~\cite{BENT}. However, identifying such examples and verifying their normality is a non-trivial research task, as the complexity of verification largely depends on the approach used for enumerating affine subspaces~\cite{CHARPIN}.
	
	A highly efficient recursive algorithm for checking normality, based on extending small flats where a function is constant to larger flats with the same property, was introduced in~\cite{BENT}. Using this method, the first non-normal and non-weakly-normal bent functions were identified in dimension $m=14$~\cite{BENT,Leander2006}. These functions are monomials over finite fields of the form $x \mapsto \operatorname{trace}(\alpha x^d)$, where $\alpha \in \mathbb{F}_{2^m}$ and $d\in \mathbb{N}$ are suitably chosen~\cite{BENT,Leander2006}. Later, additional examples were constructed in dimensions $m=12,10,14$; see~\cite{Leander2006,LeanderMcGuire2009}.
	
	This leaves dimension 8 as the only unresolved case. Notably, not all bent functions in this dimension require checking: quadratic bent functions are already known to be normal, and cubic bent functions in dimension 8 were shown to be normal in~\cite{CHARPIN}. Thus, the problem reduces to analyzing bent functions of maximum algebraic degree 4, which was posed as an open problem by Charpin~\cite[Open problem 5]{CHARPIN}.
	
	The case $m = 8$ remained unexplored for a long time until recently, when the first example of a non-normal but weakly normal bent function in this dimension was discovered~\cite{POLUJAN_2025_DAM}. However, finding a non-weakly-normal bent function (provided it exists) was left as an open problem in~\cite{POLUJAN_2025_DAM}. The challenge with this case, compared to smaller dimensions, is based on the fact that the classification of 8-variable bent functions remains an open problem that has yet to be resolved\footnote{The classification of bent functions in dimension 8 is an ongoing computational project of the current authors.}.
	
	In this paper, we propose a computational approach for investigating the (weak) normality of all 8-bit bent functions without classifying all of them. The main result of our paper can be stated as follows:
	$$\emph{All bent functions in 8 variables are normal or  weakly normal.}$$
	
	Besides this result, we introduce the concept of \emph{relative degree} for Boolean functions. This concept can be seen as a generalization of (weak) normality, providing a more detailed characterization of a given function's behaviour on affine flats. Furthermore, we investigate the relative degree of Boolean functions (not limited to bent functions) for up to $m = 8$ variables, including all functions of degree four in 7 variables and all cubic functions in 8 variables.
	
	The paper is organized in the following way. In Section~\ref{sec 2: intro}, we give the necessary background on Boolean functions that will be used throughout the whole article. Section~\ref{sec 3: relative degree} introduces the concept of the relative degree for Boolean (not necessarily bent) functions and explores its possible values for functions up to 7 variables, as well as for cubic functions in 8 variables. In Section~\ref{sec 4: sieve}, we propose a sieving algorithm aimed at finding Boolean functions with high relative degrees. Based on this algorithm, in Section~\ref{sec 5: normality 8-bit bent}, we describe an algorithm for constructing non-weakly normal bent functions in dimension 8. Using this algorithm, we show that all bent functions in eight variables are either normal or weakly normal. Some supplementary computations related to this result are provided in Appendix~\ref{appendix}. In Section~\ref{sec 6: conclusion}, we conclude the paper and outline several open problems related to the normality of Boolean functions.
	
	%\subsection{Organization of the paper and our contribution.}
	
\section{Background on Boolean functions}\label{sec 2: intro}
Let $\fd=\{0,1\}$ be the finite field with two elements and let $\F_2^m$ be the vector space of dimension $m$ over $\F_2$. We identify a vector $x=(x_1,\ldots,x_m)\in\F_2^m$ with its integer representation $x=\sum_{i=1}^{m}x_i2^{i-1}$. A mapping $f \colon \fdm \rightarrow\fd$ is called a \emph{Boolean} or \emph{$m$-bit} function. The set of all Boolean functions on $\F_2^m$ is denoted by $\boole(m)$. Any Boolean function $f$ on $\F_2^m$ is uniquely determined by the vector $(f(0),f(1),\ldots,f(2^m-1))\in \F_2^{2^m}$, which is called the \emph{truth table} of the Boolean function $f$. The \emph{Hamming weight} $\wt(f)$ of $f\in\boole(m)$ is defined as the Hamming weight of its truth table.

Every Boolean function has a unique representation as a multivariate polynomial, called the \emph{algebraic normal form (ANF)}:
\begin{equation}\label{ANF}
	f(x_1, x_2, \ldots, x_m ) = f(x) = \sum_{S\subseteq \{1,2,\ldots, m\}} a_S X_S,
	\quad a_S\in\fd, \ {\color{black} X_S = \prod_{s\in S} x_s}.
\end{equation}
The \emph{(algebraic) degree} $\deg(f)$ of a Boolean function $f\in\boole(m)$ is the maximal cardinality of $S$ with $a_S=1$ in the algebraic normal form of $f$. The \emph{valuation} $\val(f)$ of a function $f\in\boole(m)$ is the minimal cardinality of $S$ with $a_S=1$. In this paper, we conventionally fix the degree of a constant function to zero. We call Boolean functions $f\in B(m)$ with $\deg(f)=1$ \emph{affine}, with $\deg(f)=2$ \emph{quadratic}, with $\deg(f)=3$ \emph{cubic} and with $\deg(f)=4$ \emph{quartic}.

We define $\boole(s, t, m)$ as the set of Boolean functions in $m$ variables of valuation greater than or equal to $s$ and of algebraic degree less than or equal to $t$ along with the constant zero function, i.e.,  
$$
\boole(s, t, m) = \{f \in \boole(m) \colon \val(f) \geq s \text{ and } \deg(f) \leq t\} \cup \{0\}.$$
Note that $\boole(s,t,m)=\{0\}$ whenever $s>t$. It is clear that the sets $\boole(s, t, m)$ are vector spaces.

The affine general linear group of $\fdm$, denoted by $\aglm$, acts naturally over all these spaces, therefore it is essential to introduce the following equivalence relations. A system of representatives of $\boole(s,t,m)$ under the action of $\aglm$ is denoted by $\tildeboole(s,t,m)$. We will call Boolean functions $f,f'\in B(m)$ \emph{affine equivalent} if $f'=f\circ A$ for some $A\in\aglm$. A slightly more general notion of equivalence can be defined as follows. We say that two Boolean functions $f,f'\in B(m)$ are \emph{extended-affine equivalent (EA-equivalent)}, if $f'(x)=f(A(x))+a(x)$ holds for all $x\in\F_2^m$, where $A\in\aglm$ and $a$ is an \emph{affine} Boolean function on $\F_2^m$. This equivalence relation is especially useful in the study of bent functions since it trivially preserves the bent property, which is defined in the following way.

\begin{definition}
	Let $m$ be even. A Boolean function $f\in B(m)$ is called \emph{bent} if for all non-zero $a\in\F_2^m$ and for all $b\in\F_2$ the equation
	$f(x+a)+f(x)=b$ has $2^{n-1}$ solutions $x\in\F_2^n$.
\end{definition}

In the past few decades, many infinite families of bent functions have been introduced; see, e.g.,~\cite{Carlet2021,CM_Four_Decades_2016, Mesnager2016}. Proving that given infinite families of bent functions are (not) EA-equivalent is a well-known difficult problem. However, thanks to the following connection to the coding theory, it is possible to solve equivalence problems computationally for Boolean functions in a small number of variables (e.g., with the help of Magma~\cite{MAGMA:1997}). For a Boolean function $f\in B(m)$, define the linear code $\mathcal{C}_f$ as the row-space generated by the following matrix
\begin{equation*}
	{\bf G}_f= \left(\begin{array}{cccc}
		{1}&  {1}&
		{\cdots}&{1}\\[.5em]
		{x_1} & {x_2} & 
		{\cdots}& {x_{2^n}}\\[.5em]
		{f(x_1)} & {f(x_2)} 
		&\cdots & {f(x_{2^n})}\end{array}\right),%_{x\in\F_2^n}.
\end{equation*}
where all $x_i\in\F_2^m$. It is well-known that Boolean functions $f$ and $f'$ on $\F_2^m$ are EA-equivalent if and only if the linear codes $\mathcal{C}_f$ and $\mathcal{C}_{f'}$ are permutation equivalent (i.e., they are equal up to a permutation of the columns); see~\cite{Budaghyan2010CCZ_EA,EdelP09}. We provide an example demonstrating how to verify the EA-equivalence of Boolean functions using Magma.

\begin{example}
	Define the following three bent functions from the Maiorana-McFarland bent class in $n=6$ variables as follows:
	\begin{equation*}
		\begin{split}
			f(x)&=x_1x_4 + x_2x_5 + x_3x_6, \\
			g(x)&=x_1x_4 + x_2x_5 + x_3x_6 + x_1 x_2 x_3, \\
			h(x)&=g(xA+b)+x_6+1,
		\end{split}
	\end{equation*}	
	where the affine non-degenerate mapping $x=(x_1,x_2,x_3,x_4,x_5,x_6)\in\F_2^6 \mapsto xA+b$ is defined by
	$$x\mapsto \begin{pmatrix}
	    1 + x_2 + x_4\\
            x_1 + x_2 + x_4 + x_6\\
            x_2 + x_4 + x_6\\
            1 + x_1 + x_2 + x_5\\
            x_1 + x_2\\
            1 + x_2 + x_3 + x_5
	\end{pmatrix}^T.$$
	%where $A=\scalebox{0.75}{$\left(
		%\begin{array}{cccccc}
		%	0 & 1 & 0 & 1 & 1 & 0 \\
		%	1 & 1 & 1 & 1 & 1 & 1 \\
		%	0 & 0 & 0 & 0 & 0 & 1 \\
		%	1 & 1 & 1 & 0 & 0 & 0 \\
		%	0 & 0 & 0 & 1 & 0 & 1 \\
		%	0 & 1 & 1 & 0 & 0 & 0 \\
		%\end{array}
		%\right)$}$ is invertible over $\F_2$ and $b = (1, 0, 0, 1, 0, 1)$. 
	By construction, the functions $g$ and $h$ are EA-equivalent and are cubic, i.e., $\deg(g)=\deg(h)=3$. Since $f$ is quadratic, i.e., $\deg(f)=2$, we have that both $g$ and $h$ are EA-inequivalent to $f$. This fact can be verified using the following transcript of a Magma session, which can be run using the \href{https://magma.maths.usyd.edu.au/calc/}{Magma calculator}.
\end{example}

%//Type ? for help.  Type <Ctrl>-D to quit.
%\lstset{
%	caption=Checking EA-equivalence with Magma, 
%	basicstyle=\footnotesize, frame=tb,
%	xleftmargin=.01\textwidth, xrightmargin=.01\textwidth
%}
%\numberwithin{lstlisting}{section}
%\lstset{
%		caption=Checking EA-equivalence with Magma, 
%		basicstyle=\footnotesize, frame = none,
%		xleftmargin=.01\textwidth, xrightmargin=.01\textwidth
%	}
%\begin{lstlisting}[language=Magma,style=mystyle,label=lst] 
%	//Magma V2.25-3
%	n:=6;
%	R<x1,x2,x3,x4,x5,x6>:=PolynomialRing(GF(2),n);
%	F2n:=VectorSpace(GF(2), n);
%	F2nElements:=[Eltseq(v): v in F2n];
%	
%	f:=x1*x4 + x2*x5 + x3*x6;
%	g:=x1*x4 + x2*x5 + x3*x6 + x1*x2*x3;
%	h:=x1*x2 + x2*x3 + x2*x4 + x3*x4 + x5 + x6 + x2*x6
%	+x1*x2*x6 + x3*x6 + x4*x6 + x1*x4*x6 + x5*x6;
%	
%	function ComputeLinearCode(f)
%	Mf:=[[1] cat x cat [Evaluate(f,x)]: x in F2nElements];
%	Gf:=Transpose(Matrix(GF(2),2^n,n+2,Mf));
%	return LinearCode(Gf);
%	end function;
%	
%	function IsEAEquivalent(f,g)
%	Lf:=ComputeLinearCode(f);
%	Lg:=ComputeLinearCode(g);
%	return IsEquivalent(Lf,Lg: AutomorphismGroups:="Both");
%	end function;	
%	
%	ans:=IsEAEquivalent(f,g); print ans;
%	// false
%	
%	ans:=IsEAEquivalent(g,h); print ans;
%	// true
%\end{lstlisting}

%\numberwithin{myListing}{section}
%\begin{myListing}
%    \caption{Checking EA-equivalence with Magma}
    \begin{center}
        \includegraphics[scale=1]{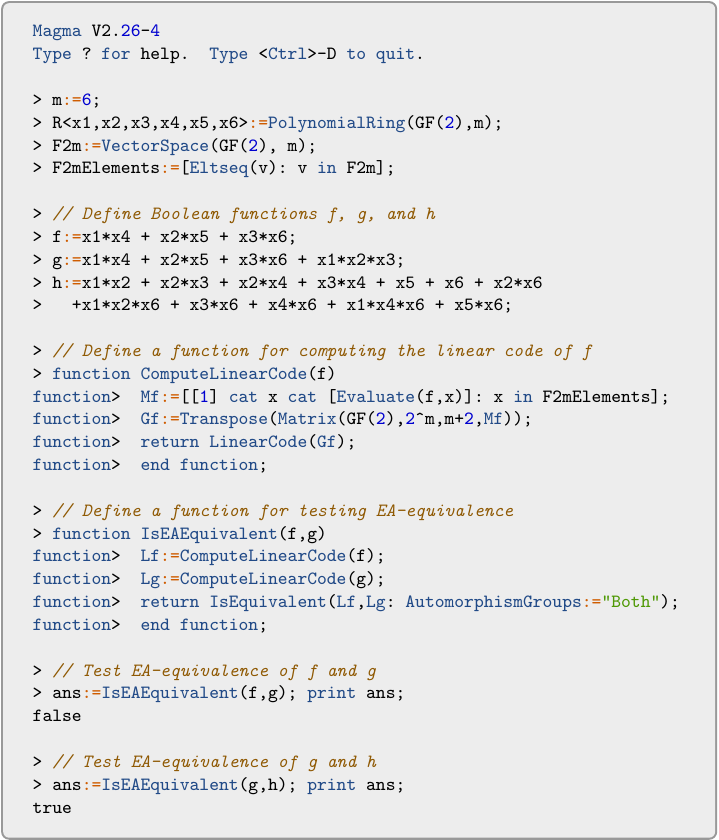}
    \end{center}
%    \label{listing: Magma}
%\end{myListing}

    \paragraph{Normality of Boolean functions.} Let $V$ be a linear subspace of $\fdm$. We denote by $\B(V)$ the space of Boolean functions from $V$ to $\fd$. For $a\in \fdm$, we denote by $f_{a,V}$ the Boolean function of $\B(V)$ defined by $v\in V \mapsto f(a+v)$, which describes the restriction of $f$ on the affine space (flat) $V+a$. With these definitions, we present the notion of normality, which was originally introduced by Dobbertin~\cite{DobbertinFSE1994} for bent functions. We follow the terminology of Charpin~\cite{CHARPIN} that applies to Boolean (not necessarily bent) functions.
	\begin{definition}
		A Boolean function $f\in\boole(m)$ is said to be \emph{normal} if there exists an affine subspace $V+a$ of $\fdm$ with dimension $\lceil m/2 \rceil$, on which $f_{a,V}$ is constant, i.e., $\deg(f_{a,V})=0$. If such a flat does not exist, $f$ is called \emph{non-normal}. A function $f\in\boole(m)$ is called \emph{weakly normal} if $f_{a,V}$ is affine, and not constant, i.e., $\deg(f_{a,V})=1$, on an affine subspace $V+a$ of dimension  $\lceil m/2 \rceil$. If such a flat does not exist, $f$ is called \emph{non-weakly-normal}. %\textbf{\color{red} Phil: at this point, f normal $\to$ non weakly. Sasha: I will double check the consistency with the paper of Pascale / our Eq.~\eqref{eq: normality types}.  }
	\end{definition}
	
	It is clear that affine equivalence preserves normality and weak normality. In the following example, we illustrate in more detail these two notions.
	
	\begin{example}\label{ex: 1}
	    Consider the following Boolean function $g\in\mathcal{B}_5$, which is defined by its ANF as follows
    \begin{equation}\label{ex: near-bent}
        g(x)=x_1 x_4 + x_2 x_4 + x_3 x_4 + x_2 x_3 x_4 + x_2 x_5 + x_3 x_5 + x_1 x_3 x_5.
    \end{equation}
    First, we show that $g$ is weakly normal. Define the flat $V_1+a_1$ in $\F_2^5$ as
    $$ V_1+a_1=\{ xA_1\colon x\in\F_2^3 \}+a_1,$$
    where the $3\times5$ full-rank matrix $A_1$ and the vector $a_1$ are given by 
    $$A_1=\left(
\begin{array}{ccccc}
 0 & 0 & 1 & 0 & 0 \\
 0 & 0 & 0 & 1 & 0 \\
 0 & 0 & 0 & 0 & 1 \\
\end{array}
\right)\quad\mbox{and}\quad a_1=(1, 1, 0, 0, 1).$$
 In turn, the restriction of $g$ on $V_1+a_1$ is given by $g_{a_1,V_1}(x)=g(xA_1+a_1)=1 + x_3$, for all $x\in\F_2^3$. Since $\deg(g_{a_1,V_1})=1$ for the flat $V_1+a_1$, the function $g$ is weakly normal. Now, define the flat $V_2+a_2$ in $\F_2^5$ as
    $$ V_2+a_2=\{ xA_2\colon x\in\F_2^3 \}+a_2,$$
    where the $3\times5$ full-rank matrix $A_2$ and the vector $a_2$ are given by 
        $$A_2=\left(
\begin{array}{ccccc}
 1 & 0 & 0 & 1 & 0 \\
 0 & 1 & 0 & 1 & 1 \\
 0 & 0 & 1 & 0 & 0 \\
\end{array}
\right)\quad\mbox{and}\quad a_2=(1, 0, 1, 0, 1).$$
 The restriction of $g$ on $V_2+a_2$ is given by $g_{a_2,V_2}(x)=g(xA_2+a_2)=0$, for all $x\in\F_2^3$. Since $\deg(g_{a_2,V_2})=0$ for the flat $V_2+a_2$, the function $g$ is normal.
\end{example}

  We note that Dickson's classification of quadratic functions enables a complete analysis of their normality. In the following remark, we summarize the analysis of the normality of quadratic functions presented in detail in~\cite[Appendix]{CHARPIN}.
	\begin{remark}
		Suppose $f\in B(m)$ has degree 2. From the classification of quadratic forms~\cite[pp. 438–442]{MacWilliamsSloane}, we know that if $f$ is balanced, that is, $\wt(f)=2^{m-1}$, then $f$ is affine equivalent to $$Q_k(x)=x_1 x_2+x_3 x_4+\ldots+x_{2 k-1} x_{2 k}+x_{2 k+1},$$ for some $k \leq(m-1) / 2$. If $f$ is not balanced, then $f$ is affine equivalent to $$Q_{k,b}(x)=x_1 x_2+x_3 x_4+\ldots+x_{2 k-1} x_{2 k}+b,$$ for some $k \leq m / 2$ and $b$ in $\F_2$. More precisely, if $\wt(f)<2^{m-1}$, then $b=0$, if $\wt(f)>2^{m-1}$, then $b=1$. Consequently, if $m$ is even, all quadratic functions are normal since the functions $Q_k$ and $Q_{k,b}$ on $\F_2^m$ are constant on the vector space $V=\langle e_1, e_3,\ldots, e_{m/2} \rangle$ of dimension $m/2$, where $e_k$ is the $k$-th unit vector. The same argument shows that for $m$ odd, all quadratic functions that are affine equivalent to $Q_k$ and $Q_{k,b}$ on $\F_2^m$, where $k<(m-1)/2$ are normal since they are constant on the vector space $V=\left\langle e_1, e_3,\ldots, e_{\frac{m-1}{2}},e_{\frac{m-1}{2}+1} \right\rangle$ of dimension $\lceil m/2 \rceil$. If $k=(m-1)/2$, the function $Q_k$ restricted to the same subspace $V$ is linear and hence non-normal (but weakly normal).
	\end{remark}
	
	Note that determining (weak) normality theoretically is more challenging for non-quadratic functions, as their classification is not known for all numbers of variables. The next section explores this issue in detail, concentrating on Boolean functions in a small number of variables.
	
	%	While for quadratic functions the normality is easy to analyse thanks to their classification, for arbitrary functions, whose structure is not known. this tasks becomes challenging. Say that with this notion, investigation of normaity of Boolean (not necessary bent) functions, attracted a lot of attention (refer to Claude's Book / computational results).
	
	\section{Relative degree --- A generalization of normality}\label{sec 3: relative degree}

	Since the introduction of this concept, normality has been studied not only in the context of bent functions but also for Boolean functions. For a recent and thorough survey on the topic, we refer the reader to~\cite[Section 3.1.8]{Carlet2021}. A common research question in this direction is whether a given function or even a class of functions is normal or weakly normal. However, we want to know more: how exactly are these functions non-normal? To be more precise, what is the maximum degree of a given Boolean function while restricted to all flats of a certain fixed dimension? To answer the latter question, we introduce the following generalization of normality.
	
	\begin{definition}
		The \emph{relative degree} of a Boolean function $f\in B(m)$ with respect to an affine space $a+V$ of $\F_2^m$ is the degree of $f_{a,V}$, which is denoted by $\degV(a,V,f)=\deg(f_{a,V})$. In this context, we set conventionally $\degV(a,V,f)=0$ if $f_{a,V}=0$.
		We define the \emph{$r$-degree} of $f\in\boole(m)$ as the minimal relative degree of $f$ 
		for all affine spaces $a+V$ of $\F_2^m$, where $\dim (V)=r$:
		$$\deg_r(f)=\min\{\degV (a,V,f) \mid \dim V=r \;  \text{and}\;  a\in \fdm \}.$$
	\end{definition}
	
	We note that similarly to normality and weak normality, the $r$-degree of a Boolean function is invariant under affine equivalence.
	
	\begin{remark}
		It is interesting to note that the concept of $r$-degree appears independently in several recent studies. In~\cite{STABILITY}, the authors are 
		interested in Boolean functions preserving their degree by hyperplane restriction. In~\cite{HAUGLAND}, the $r$-degree of $f$ is denoted by $\alpha(f, r)$, and the parameter $g(m,r) := \max_{f} \alpha(f, r)$ is studied. In this paper, for a given $k$, we aim to bound the maximal $r$-degree of a Boolean function of degree $k$.
	\end{remark}

	For a given Boolean function $f\in B(m)$, one can use the values of $\lceil m/2 \rceil$-degree in order to determine normality and weak normality as follows:
	\begin{equation}\label{eq: normality types}
		\deg_{\lceil m/2 \rceil}(f)=\begin{cases}
			0, & \text{$f$ is normal};\\
			1,& \text{$f$ is non-normal and weakly normal};\\
			\geq 2,& \text{$f$ is non-normal and non-weakly-normal}.\\
		\end{cases}
	\end{equation}
	Based on this correspondence, we refer to Boolean functions that are both non-normal and non-weakly normal as abnormal.
	\begin{definition}
		We call a Boolean function $f\in B(m)$ with $\deg_{\lceil m/2 \rceil}(f)\ge 2$ \emph{abnormal}.
	\end{definition}
	 	
	To determine whether a function $f \in \boole(m)$ is abnormal, we need to verify that the function $f$ is constant or affine on a flat of dimension $\lceil m/2 \rceil$. This is equivalent to checking that the function is constant on two parallel flats, i.e., distinct translates of the same linear subspace of dimension $\lceil m/2 \rceil - 1$. This observation motivates Algorithm~\ref{ABNORMAL} for checking the abnormality of a given Boolean function.
	
	\begin{algorithm}
		\caption{$\text{isABNORMAL}(f)$}
		\begin{algorithmic}[1]
			\Require A Boolean function $f$ in $\boole( m )$.
			\Ensure If $f$ is abnormal return \textsc{true}, else return \textsc{false}.
			
			\For{all  vector spaces  V of dimension $\lceil m/2 \rceil-1$  }
			\State  count := 0 
			\For{all  $a$  in $\fdm / V$  }
			\If {$f$ is constant on $a+V$}
			\State  count++
			\If {count = 2 } 
			\State  return \textsc{false} 
			\EndIf
			\EndIf
			
			\EndFor
			\EndFor
			\State \Return \textsc{true} 
		\end{algorithmic}
		\label{ABNORMAL}
	\end{algorithm}
	
	\begin{remark}
		There are $99\,471$ 3-dimensional linear subspaces in $\F^8_2$ and $3\,108\,960$ $\approx 2^{21.6}$ 3-dimensional affine spaces, meaning that (the naive) Algorithm~\ref{ABNORMAL} is sufficiently efficient for $m\leq 8$, the dimension of our interest. For a more efficient algorithm, see~\cite{BENT}.
	\end{remark}
	
	In the remaining part of this section, we investigate the following combinatorial parameter that we use to analyze the type of normality for Boolean functions in a small number of variables in the sense of Eq.~\eqref{eq: normality types}. For integers $r\leq m$ and $k\leq m$, we denote by
	\begin{equation}\label{eq: Drkm eq}
		D^\dag_r( k, m) = \max \{\deg_r( f ) \colon f\in B(m) \mbox{ and }  \deg(f)  = k \}
	\end{equation}
	the maximum $r$-degree of all $m$-bit Boolean functions of degree \emph{equal} to $k$. With this value, one can study the normality of Boolean functions as follows. For instance, when $D^\dag_{\lceil m/2 \rceil}(k, m) = 0$, it indicates that all $m$-bit functions of degree $k$ are normal. If $D^\dag_{\lceil m/2 \rceil}(k, m) = 1$, it implies that all $m$-bit functions of degree $k$ are either normal or weakly normal. Lastly, a value of $D^\dag_{\lceil m/2 \rceil}(k, m) \geq 2$ indicates the existence of abnormal $m$-bit functions of degree $k$.
	
	Our next objective is to determine the values $D^\dag_r( k, m)$ for different values of $r$ and $k$ when the number of variables is $m=5,6$ or $7$. These findings will enable us to verify certain known results regarding the normality of Boolean functions and to extend our study of normality to the case $m=8$. To do so, we introduce the following combinatorial parameter
	\begin{equation}\label{eq: Drkm leq}
		D_r( k, m) = \max \{\deg_r( f ) \colon f\in B(m) \mbox{ and }  \deg(f) \leq  k \},
	\end{equation}
	which is the maximum $r$-degree of all $m$-bit Boolean functions of degree \emph{less or equal} to $k$. From the definitions of $D^\dag_r( k, m)$ and $D_r( k, m)$, it is also clear that 
	\begin{equation}\label{eq: Drkm leq computing}
		D_r( k, m) = \max \{ D^\dag_r( l, m) \colon 1\le l\le k  \}.
	\end{equation}
	In the following lemma, we summarize the connections between the values $D^\dag_r( k, m)$ and $D_r( k, m)$ and prove the non-decreasing property of the mapping $m\mapsto D_r( k, m)$.
	
	\begin{lemma}%[new monotony]
		\label{MONO}
	     The relative degree satisfies: 
		\begin{itemize}
		    \item[(1)] For all integers $0 \leq k\leq m$:
		    $$ D^\dag_r( k, m)\le D_r( k, m),$$
	    \item[(2)] For all integers $0 \leq m \leq k \leq m'$:
		        $$D_r(k, m') \leq  D_r( \min\{k,m\} , m).$$
		 \end{itemize}
	\end{lemma}
%	\begin{lemma}[ old monotony]
%		\label{OLDMONO}
%		For all integers $0 \leq k\leq m\leq m'$, the following inequalities hold:
%		\begin{itemize}
%			\item[(1)] $D^\dag_r( k, m)\le D_r( k, m)$, 
%			\item[(2)] $D_r(k, m') \leq  D_r(k, m)$.
%		\end{itemize}
%	\end{lemma}
	\begin{proof}
		The first claim follows from Eq.~\eqref{eq: Drkm leq computing}.
		 %the definitions of $D^\dag_r( k, m)$ and $D_r( k, m)$ given in Eqs.~\eqref{eq: Drkm eq} and~\eqref{eq: Drkm leq}, respectively, and the fact that $\{f\in B(m) \colon  \deg(f)  = k \}\subset \{ f \in B(m) \colon  \deg(f)  \le k \}$. 
		 To prove the second claim, we observe that the restriction of $f'\in B(m')$ with $\deg(f')\le k$ to a vector space of dimension $m$ can be identified with a function $f\in B(m)$ with $\deg(f)\le k$ whose $r$-degree is less or equal to $D_r(k, m)$.
	\end{proof}
	% (i.e., elements of the spaces $\tildeboole(s,t,m)$)
	\begin{remark}
		In general, determining $D_r(k, m)$ is a challenging problem, as we will now explain. The affine general linear group acts on the set of affine spaces of a given dimension, preserving the degree of functions. Consequently, the relative degrees are affine invariants of $\boole(m)$. This allows us to utilize the classification of Boolean functions to determine some of these parameters. Each linear subspace of dimension $r$ ($r$-space) of $\mathbb{F}_2^m$ has $2^{m-r}$ cosets, and the number of $r$-spaces is given by the Gaussian binomial coefficient $\binogauss(m, r)$. Considering the computational cost $r2^{r}$ of determining the degree of a Boolean function in $r$ variables~\cite[pp. 33-34]{Carlet2021}, the work factor required to compute the $r$-degree of Boolean functions in $\boole(s, t, m)$ through brute force is as follows:
	\begin{equation}\label{eq: working factor}
			W( r, s, t ,m) =  |\tildeboole(s,t,m)| \times  2^{m-r} \binogauss(m,r)\times r2^{r},\quad   
		\binogauss(m,r)=\prod_{i=0}^{r-1} \frac{2^m - 2^{i}}{2^r - 2^i}.
	\end{equation}
	\end{remark}
		%In~\cite{project}, the first two authors of this work provided the classification of the spaces $\boole(s,t,m)$ under the action of $\aglm$ and gave the cardinalities of systems of representatives $\tildeboole(s,t,m)$. For the reader's convenience, in Table~\ref{CLASS}, we give the cardinalities of spaces $\tildeboole(s,t,m)$ for some spaces $\boole(s,t,m)$ that we use later to determine the values of $D^\dag_r( k, m)$. These values will also be useful for estimating the working factor given in Eq.~\eqref{eq: working factor}, in each particular case.
		In~\cite{project}, the first two authors of this article classified some of the spaces $\boole(s,t,m)$ under $\aglm$ and determined the cardinalities of the equivalence classes $\tildeboole(s,t,m)$. Table~\ref{CLASS} lists these cardinalities of selected spaces $\boole(s,t,m)$, used later to compute $D^\dag_r(k, m)$ and estimate the working factor in Eq.~\eqref{eq: working factor} for each case.
	
	\begin{table}[h]
		\small
		\caption{Some spaces $\boole(s,t,m)$ and the values $|\tildeboole(s,t,m)|$}
		\begin{tabular}{|c|c|c|c|c|}
			\hline 
			$\boole(s,t,m)$
			&$\bst(1,5,5)$
			& $\bst(1,6,6)$
			&$\bst(1,3,7)$  
			&$\bst(2,4,7)$ \\
			%$\bst(3,4,7)$
			%$\bst(4,7,7)$
			%$\bst(5,7,7)$ \\
			\hline
			$|\tildeboole(s,t,m)|$
			& $206$ 
			&$7\,888\,299$
			&$1\,890$ 
			&$118\,140\,881\,980$
			%&$68\,443$ 
			%&$3\,486$
			%&$12$
			\\
			\hline
		\end{tabular}
		\begin{tabular}{|c|c|c|}
			\hline 
			 $\bst(3,4,7)$
			 &$\bst(4,7,7)$
			 &$\bst(5,7,7)$\\
			\hline
			$68\,443$ 
			&$3\,486$
			&$12$ \\
			\hline
		\end{tabular}
		\label{CLASS}
	\end{table}
	
Since the classification of the spaces $\boole(1,5,5)$ and $\boole(1,6,6)$ is known~\cite{bfapaper2022}, we can use the systems of representatives of equivalence classes from~\cite{project} to compute the 
values $D^\dag_r( k, 5)$ and $D^\dag_r( k, 6)$, which we present in Tables~\ref{DR5} and~\ref{DR6}, respectively. Note that for $m=6$, $r=4$ and $|\tildeboole(1,6,6)|=7\,888\,299 \approx 2^{23}$. In this way, from Eq.~\eqref{eq: working factor}, we deduce that the work factor is about $2^{40}$. Using an 80-core computer, we compute all values in Tables~\ref{DR5} and~\ref{DR6} in a matter of minutes. From the values $D^\dag_{\lceil m/2 \rceil}(k, m)$, where $m=5,6$ and $1\le k\le m$, we deduce:

	\begin{theorem}\label{NORMAL6}
		 All Boolean functions in $\boole(5)$ are normal or weakly normal. All Boolean functions in $\boole(6)$ are normal.
	\end{theorem}
	Note that the latter claim was proven in \cite{dubuc}; thus, we confirm the correctness of this result. To proceed with the values $D^\dag_r( k, 7)$, we first fill the first three columns of Table~\ref{DR7} using the classification of $\boole(1,3,7)$ from~\cite{bfapaper2022,project}.  Tables~\ref{DR5} and~\ref{DR6} in a matter of minutes. From the value $D^\dag_{4}(3, 7)$, we deduce:
	
	\begin{theorem}\label{CUBICFACT7}	  
		All Boolean cubics in  $\boole(7)$ are normal or weakly normal.
	\end{theorem}
	
	\begin{table}[h!]
		\caption{The values of $D_r^\dag(k, m)$, for  $m=5,6,7$}
		\centering
		\begin{minipage}{\textwidth}
			\begin{subtable}[t]{0.45\textwidth}
				\centering
				\caption{$D^\dag_r( k, 5)$}
				\begin{tabular}{|c|ccccc|}
					\hline
					\diagbox[width=18pt,height=18pt]{$r$}{$k$}  &1 &2 &3 &4 &5\\
					\hline
					\hline
					4 &0&2 &2 &3 & 2\\
					%\rowcolor[gray]{.8}
					3 &0 &1 &1 &1 & 1\\
					
					2 &0 &0 &0 &0 & 0\\
					\hline
				\end{tabular}
				\label{DR5}
			\end{subtable}%
			\hfill
			\begin{subtable}[t]{0.45\textwidth}
				\centering
					\caption{$D^\dag_r( k, 6)$}
				\begin{tabular}{|c|cccccc|}
					\hline
					\diagbox[width=18pt,height=18pt]{$r$}{$k$}  &1 &2 &3 &4 &5 & 6 \\
					\hline
					\hline
					5 & 0 & 2&	3& 	4& 	3&	4 \\
					4 &0 &	2&	2&	2& 	2&	2\\
					3 &0&	0&	0 &	0& 	0&	0\\
					%\rowcolor[gray]{.8}
					2 &0&	0&	0&	0 &	0&	0\\
					\hline
				\end{tabular}
				\label{DR6}
			\end{subtable}
		\end{minipage}
		
		\vspace{0.1cm}
		
		\begin{minipage}{\textwidth}
			\centering
			\begin{subtable}[t]{0.9\textwidth}
				\centering
				\caption{$D^\dag_r( k, 7)$}
				\def\rind{\textcolor{black}{1 or 2}}
				\def\rund{\textcolor{black}{3 or 4}}
				\def\rand#1{\textcolor{black}{$\geq #1$}}
				\def\rond#1{\textcolor{red}{$#1$}}
				\def\rend#1{\textcolor{blue}{$#1$}}
				\begin{tabular}{|c|ccc|cccc|}
					\hline
					\diagbox[width=18pt,height=18pt]{$r$}{$k$} 	&1 	&2 	&3 	&4 	&5 	&6	&7\\
					\hline
					\hline
					6	&0	&2 	&3	&\rand 4 	&\rand 4 	&5	    &\rand 4\\
					5	&0	&2	&3	&\rund  &\rund  	&\rund	&\rand 3\\
					4	&0	&1	&1	&  \dbox{2} &\rind	&\fbox{2}	&\rind \\[0.5mm] \hline
					3	&0	&0	&0	&0	&0 	&0	&0\\
					2	&0	&0	&0	&0  &0 	&0	&0\\
					\hline
				\end{tabular}
				\label{DR7}
			\end{subtable}
		\end{minipage}
		\label{D-values} 
	\end{table}
			
	In~\cite{dubuc}, it is proved that all 8-bit cubic \emph{bent} functions are normal. With our results, we can generalize this statement even further:
	
	\begin{theorem}
		All Boolean cubics in $\boole(8)$ are normal or weakly normal.
	\end{theorem}
	\begin{proof}
		By Theorem~\ref{CUBICFACT7}, all Boolean cubics in $\boole(7)$ are normal or weakly normal. Using Eq.~\eqref{eq: Drkm leq computing} and applying Lemma~\ref{MONO}, we get that $$D^\dag_{4}(3, 8)\le D_{4}(3, 7)=\max\{ D^\dag_{4}(1, 7),D^\dag_{4}(2, 7),D^\dag_{4}(3, 7)\}=1,$$
		which completes the proof.
	\end{proof}
	
	In the rest of this section, we proceed with filling the remaining values of Table~\ref{DR7}. We notice that the value $D^\dag_4(6, 7) = 2$ (given in the non-dashed box) was recently determined in \cite{HAUGLAND} using a randomized approach. The value $D^\dag_4(4, 7) = 2$ (given in the dashed box), will be computed separately in Section~\ref{sec 4: sieve} due to its importance for the analysis of normality of 8-bit bent functions that we perform in Section~\ref{sec 5: normality 8-bit bent}. Zero values in the rows $r=2,3$ are obtained directly by applying Lemma~\ref{MONO} to the values in Table~\ref{DR6}.
	
	Since  $|\tildeboole(2,4,7)| = 118\,140\,881\,980 \approx 2^{37}$, we can not reasonably use brute force in dimension $7$ to fill the remaining values $D^\dag_r(k, 7)$ in the columns $k=4,5,6,7$ of Table~\ref{DR7}.  To determine these values, we apply Lemma~\ref{MONO} for upper bounds and perform various computations, including classifications and random searches, to establish the lower bounds; for supplementary computations, see the project's web page~\cite{project_normal}. 
    %Using the values $D^\dag_4( k, 6)$ for $1\le k \le 6$ given in Table~\ref{DR6}, we deduce that $D_4( k, 6 ) = 2$ using Eq.~\eqref{eq: Drkm leq computing}. By Lemma~\ref{MONO}, for all $k\in \{4,5,6,7\}$ it holds that $D^\dag_4( k, 7) \leq D_4( k, 6 ) = 2$. 
From the line $r=4$ in Table~\ref{DR6}, we deduce that for all $k\in \{4,5,6,7\}$ it holds that $1 \le D^\dag_4( k, 7) \leq D_4( k, 6 ) = 2$. The lower bounds in the right part of the table for $D^\dag_r(k, 7)$ were obtained by adding a random cubic function to the elements of $\tildeboole(4,7,7)$.
	
	\begin{example}
		The 5-relative degree of the function $f\in \tildeboole(3,4,7)$ defined by
		\begin{equation*}
			\begin{split}
				f(x) &= x_2x_3x_4x_5+x_1x_2x_3x_6+x_1x_4x_6+x_3x_4x_5x_6+x_2x_3x_7+x_4x_5x_7\\
				&+x_3x_4x_5x_7+x_1x_3x_6x_7+x_3x_4x_6x_7+x_1x_5x_6x_7
			\end{split}
		\end{equation*}
		is equal to 3. This provides the lower bound $3\leq D^\dag_5(4,7)$ and the upper bound  $D^\dag_5(4,7)\leq D_5(4,6) = 4$.
	\end{example}
	
	\begin{example}
		Among the 12 members of $\tildeboole(5,7,7)$, all have 6-relative degree 0, except for
		$$h(x)=x_1 x_2 x_3 x_4 x_5 x_6 + x_2 x_3 x_4 x_5 x_7 + x_1 x_3 x_4 x_6 x_7 + x_1 x_2 x_5 x_6 x_7,$$
		%$$h := abcdef+bcdeg+acdfg+abefg$$
		which satisfies $\deg_6(h) = 5$. This establishes the value $D^\dag_6( 6, 7)=5$ in Table~\ref{DR7}.
	\end{example}
	
    At the BFA 2024 conference, we proposed in~\cite{GLP_BFA_2024} the conjecture that all 7-bit quartics are normal or weakly normal. In the following section, we disprove this conjecture by showing the existence of abnormal quartics in $\boole(7)$.
	
	%{\color{blue}\textbf{Sasha: To summarize this section. To me everything is more or less fine except the fact that I feel there is something missing in Lemma~\ref{MONO}. Somehow, I don't see how we get the values in the last column. I also think that for all lower bounds we need to provide concrete examples that gives us the values. Probably we still need to keep 1 or 2 examples in this section and to move the others in Appendix. We also need to think whether we need all values in Table~\ref{CLASS} or if some are missing.} }
	
	\section{Relative degree and abnormality of 7-bit quartics}
	\label{sec 4: sieve}
	In the previous section, we established the bounds $1 \leq D^\dag_4(4, 7) \leq 2$. Here, we prove that the exact value of $D^\dag_4(4, 7)$ is 2, thereby confirming the existence of abnormal 7-bit quartics. Given that  $|\tildeboole(2,4,7)| = 118\,140\,881\,980 \approx 2^{37}$, a brute-force computation of the 2-degree for Boolean functions in $\tildeboole(2,4,7)$ is time-consuming. Instead, an alternative approach involves computing the 2-degree of functions of the form $f + q$, where $f \in \tildeboole(3,4,7)$ and $q \in \boole(2,2,7)$. This requires analyzing $|\tildeboole(3,4,7)| \times 2^{\binom{7}{2}} \approx 2^{16} \times 2^{21}$ combinations, resulting in a comparable computational complexity. However, as we will demonstrate, this complexity can be substantially reduced by analyzing only a subset of the $2^{21}$ quadratic forms. To proceed, we first introduce the following notation.
	
	Let  $\B_V(s, t )$ denote the space of Boolean functions defined on the vector space $V$ with valuation greater or equal to $s$ and degree less or equal to $t$. For $a\in\fdm$, we consider the \emph{restriction map} over  the quadratic functions modulo the space of affine functions on $V$: 
	%{\color{blue} \textbf{(Sasha: Why $\operatorname{res}_a$ and not $\operatorname{res}_{a,V}$??)}}
	\def\res{{\rm res}_{a,V}}
	\begin{equation}
		\begin{array}{ccc}
			\boole(2, 2, m ) & \xrightarrow{\res} & \B_V(2,2) \\
			q & \longmapsto & q_{a,V} \mod \B_V(0, 1)
		\end{array}
	\end{equation}

%	\begin{equation}
	%		\begin{CD}
		%			\boole(2, 2, m ) @>\res >> \B_V(2,2)\\
		%			q @>>>  q_{a,V} \mod \B_V(0, 1) \\
		%		\end{CD}
	%	\end{equation}
%
%	\begin{equation}
	%		\begin{tikzcd}[column sep=1cm, row sep=0.1cm]
		%			\boole(2, 2, m ) \arrow[r, "\res", shorten <=1pt, shorten >=1pt] & \B_V(2,2) \\
		%			q  \arrow[maps to]{r} & q_{a,V} \mod \B_V(0, 1)
		%		\end{tikzcd}
	%	\end{equation}
	
	\def\Ker{K_{a,V}}
	
	It is clear that the linear map $\res\colon \boole(2, 2, m )\to \B_V(2,2)$ is surjective, since any element of $\B_V(2,2)\subset \boole(2, 2, m )$ is a fixed point of $\res$. By the rank-nullity theorem, its kernel, denoted by $\Ker$, has dimension $\binom{m}{2}-\binom{\dim V}{2}$. The kernel of $\res$ plays a crucial role in speeding up the search for abnormal functions. Let $f\in\tildeboole(3,t,m)$ such that $\deg( f_{a,V}) \leq 2 $, for some flat $a+V$ of $\F_2^m$. If the quadratic part of $f_{a,V}$ is cancelled by some quadric $q_0\in \boole(2,2,m)$, then, it is also cancelled by every element of $q_0 + \Ker$.  Based on this idea, we propose Algorithm~\ref{SIEVING} that determines all quadratic forms $q$ for a given function $f$ such that $f+q$ is abnormal. It applies a sieving procedure to eliminate all the quadrics in the set $Q(f)$ such that the relative degree of $h+q$ on some affine $\lceil m/2 \rceil$-space is less or equal to one.
	
	\def\abnormal(#1){$\textsc{sieving}(#1)$}
	
	\begin{algorithm}[H]
		\caption{$\text{SIEVING}(f)$}
		\begin{algorithmic}[1]
			\Require A member $f$ of $\tildeboole(3,t,m)$.
			\Ensure The set  $Q(f)=\{ q\in \boole(2, 2, m )\colon f+q \mbox{ is abnormal} \}$.
			\State Define $Q(f):=\boole(2, 2, m )$. 
			\For{all $\lceil m/2 \rceil$-vector spaces $V$ }
			%\State $W$ := supplementary of $V$ 
			\For{all  $a$  in $\F_2^m / V$ } 
			\State $\rho \gets f_{a,V} \mod \B_V(0, 1)$
			\If{$\deg( \rho) \leq 2 $}
			\State choose $q_0\in\boole(2, 2, m )$ s.t. $\res( q_0 ) = \rho $
			\State remove all $q\in q_0 + \Ker$ from $Q(f)$ 
			\EndIf
			\EndFor
			\EndFor
			\State \Return $Q(f)$
		\end{algorithmic}
		\label{SIEVING}
	\end{algorithm}
	
		%{\color{blue} \textbf{(Sasha: Algorithm~\ref{SIEVING}, line 6: Do we have a problem here that there might exist no $q_0$ s.t. $\res(q_0)=\rho$ if $\rho$ has affine terms, since $\operatorname{Im}(\res)=\B_V(2,2)$, but not all quadratic functions on $V$?)}}
	
		We now apply Algorithm~\ref{SIEVING} to Boolean functions $f \in \tildeboole(3,4,7)$ to identify the sets of quadratic functions $Q(f)$ and, subsequently, the abnormal functions of the form $f+q$, where $q \in Q(f)$. In this case, $\dim \Ker = 21 - 6 = 15$, allowing us to eliminate $2^{15}$ quadratic forms $q \in q_0 + \Ker$ that lead to normal or weakly normal functions $f+q$, provided a single such one $q_0$ is identified. Our C implementation takes approximately 10 minutes on an 80-core machine to compute a cover set of size $39\,952$ for abnormal 7-bit quartics. It was found that $1\,915$ functions $f\in\tildeboole(3,4,7)$ yield abnormal functions $f+q$. Thus, we show:

		\begin{theorem}
			There exist abnormal quartics in $B(7)$. Consequently, $D^\dag_4(4, 7)=2$.
		\end{theorem}
%The work factor takes the form :	
%	$$
%	|\tildeboole(3,4,7)| \times  2^{3} \times  \binogauss(7,4)  \times 2^3 \times 2^{15}
%	$$
%	

	For $f \in \tildeboole(3,4,7)$, the value $ |Q(f)|$, computed using Algorithm~\ref{SIEVING}, represents the number of quadratic forms $q \in \boole(2,2,7)$ such that $f+q$ is abnormal. We note that Algorithm~\ref{SIEVING} produces 16 distinct values $|Q(f)|$. In Table~\ref{Nh-values}, we list these values along with their multiplicities.

	%l@812/tempo> xclip -o| sed 's/^/%/' 
	%1535	1	1535		
	%1	12	12		
	%1	16	16		
	%281	2	562		
	%2	2048	4096		
	%1	24	24		
	%19	3	57		
	%1	32	32		
	%1	32768	32768		
	%1	384	384		
	%48	4	192		
	%1	48	48		
	%3	6	18		
	%1	64	64		
	%18	8	144		
	%1914	35420	39952		20,8735632183908
	%66528		0	
	
	\setlength{\tabcolsep}{3pt}
	\begin{table}[h!]
		\caption{Distribution of the values $|Q(f)|$, for  $f\in\tildeboole(3,4,7)$}
		\centering
		\footnotesize
			\begin{tabular}{|c||c|c|c|c|c|c|c|c|c|c|c|c|c|c|c|c|c|c|}
		\hline
		$|Q(f)|$& 12 & 16 & 24 & 32 & 64 & 32768 & 384 & 48 & 2048 & 6 & 8 & 3 & 4 & 2 & 1 &0 \\
		\hline
		Mult. & 1 & 1 & 1 & 1 & 1 & 1 & 1 &1 & 2 & 3 & 18 & 19 & 48 & 281 & 1535& 66528\\
		\hline
	\end{tabular}
		\label{Nh-values}
	\end{table}
%	{\color{blue} \textbf{(Sasha: Should it be in Table~\ref{Nh-values} instead of $66535$ the value  $66528$? With the current value $66535$ we don't get the sum $68443$ in the multiplicity row (it is the number of elements in $\tildeboole(3,4,7)$???))}}
	
	%{\color{blue} Sasha: Do I understand correctly, that this table shows the distribution of elements from the cover set containing 39952 elements? The problem with Table~\ref{Nh-values} that I have is that if I compute the sum $12\times1+ 16\times1+24\times1+\cdots+1\times 1535+0\times66535$ I get the value 39888, but not 39952. Do I miss something?}
	
	%{\color{magenta} phil : 64 of multiplici ty 1 was missing} 
	
	In most cases, we have that $|Q(f)|=0$, indicating that most of quartics are normal or weakly normal. We also note that $1\,535$ functions $f\in\tildeboole(3,4,7)$ have a unique quadratic function $f\in B(2,2,7)$ such that $f+q$ is abnormal. Below, we give an example of such $f$ and $q$. %	in average, $N(h) = 0.58$.
	
	%{\color{blue} Sasha: I think we need one example here, like taking a concrete $h\in\tildeboole(3,4,7)$, finding a suitable quadratic $q$, s.t. $h+q$ is abnormal. Reviewers might not like that we show existence of such abnormal functions, but don't give a single concrete example.}

	\begin{example}
		Consider the following Boolean function $f\in\tildeboole(3,4,7)$ given by its algebraic normal form:
		%$$h:=abd+bcd+bce+ade+cde+bcde+adf+bcef+cdef+abg+acfg+defg$$
		\begin{equation*}
			\begin{split}
				f(x)&=x_1 x_2 x_4 + x_2 x_3 x_4 + x_2 x_3 x_5 + x_1 x_4 x_5 + x_3 x_4 x_5 + x_2 x_3 x_4 x_5 + 
				x_1 x_4 x_6\\ 
				&+ x_2 x_3 x_5 x_6 + x_3 x_4 x_5 x_6 + x_1 x_2 x_7 + x_1 x_3 x_6 x_7 + 
				x_4 x_5 x_6 x_7.
			\end{split}
		\end{equation*}
		Applying Algorithm~\ref{SIEVING}, we see that $Q(f)=\{q\}$, where
		$$q(x)=x_2 x_3 + x_1 x_5 + x_2 x_5 + x_3 x_5 + x_3 x_7 + x_5 x_7 + x_6 x_7.$$
		In this way, there exists only one quadratic function $q\in\boole(2,2,7)$ such that $f+q$ is abnormal. The abnormality of $f+q$ can be verified with Algorithm~\ref{ABNORMAL}.
		%$$bc+ae+be+ce+cg+eg+fg$$
		%$N(h) = 1$  and 
		%$$ q = bc+ae+be+ce+cg+eg+fg$$	
	\end{example}

	In the following section, we use the ideas developed in this section to study the normality of bent functions in eight variables.

	\section{Analysis of the normality of all 8-bit bent functions}\label{sec 5: normality 8-bit bent} 
	
	\def\fdmm{\F^{m-1}_2}
	\def\fdmr{\F^{m-r}_2}
	\def\feight{\F^{8}_2}
	\def\fseven{\F^{7}_2}
	\def\fsix{\F^{6}_2}
	\def\ffive{\F^{5}_2}
	\def\ftwo{\F_2}
	
	\def\walsh#1(#2){ \widehat{#1}(#2)}
	\def\sgn{\textsc{sgn}}
	
	Recall that all 6-bit functions are normal and only a few examples of abnormal bent functions have been found in dimensions greater than or equal to 10~\cite{BENT, LeanderMcGuire2009}. However, the existence of an abnormal 8-bit bent function has been an open question for 30 years. As mentioned in~\cite{KIT}, one could use a classification of 8-bit bent functions to determine whether any are abnormal. In this section, we propose a more efficient method that does not require such a classification. Specifically, we combine the sieving procedure from the previous section with techniques from~\cite{MENG}, which construct bent functions from their restrictions, to establish that all bent functions in $\boole(8)$ are either normal or weakly normal.
	
	\subsection{Decomposition of a bent function into two near-bent functions}
	In this section, we provide the necessary definitions for the procedure (described in the following section) used to prove that all 8-bit bent functions are normal or weakly normal. We begin with the notion of Walsh coefficients, several useful equalities, and a spectral characterization of Boolean bent functions based on them.
	
	\begin{definition}
		The \emph{Walsh coefficient} of $f\in \boole(m)$ at $a \in \fdm$ is defined by
		$$\walsh f (a)=\sum_{x \in \fdm} (-1)^{f(x)+a\cdot x}.$$
		The set of Walsh coefficients of $f\in \boole(m)$ is called the \emph{Walsh spectrum} of $f$. The Walsh coefficients satisfy~\emph{Parseval's identity}:
		\begin{equation}
			\label{PARSEVAL}
			\sum_{a\in \fdm} \walsh f(a)^2 = 2^{2m} 
		\end{equation}
		and \emph{Poisson's formula}:
		\begin{equation}
			\label{POSSON}
			\sum_{v\in V}(-1)^{f(v+b)+c\cdot v}=\frac{1}{|V^\perp|}\sum_{a\in V^\perp} \walsh f(a+c) (-1)^{b\cdot c},
		\end{equation}
		where $b,c \in \fdm$, $V$ is a linear subspace of $\fdm$ and $V^\perp$ is the \emph{orthogonal complement} of $V$ defined by $$V^\perp=\{ a\in\F_2^m\colon a\cdot v=0 \mbox{ for all } v\in V \}.$$
	\end{definition}
	
	The following well-known result provides a spectral characterization of bent functions through their Walsh coefficients.

	\begin{theorem}\cite{ROTHAUS1976,Dillon1974}
		Let $f\in\boole(m)$ be a Boolean function. The following statements are equivalent:
		\begin{enumerate}
			\item The function $f$ is bent on $\F_2^m$.
			\item The Walsh coefficients of $f$ satisfy $$\walsh f (a)\in \{\pm 2^{m/2}\}, \mbox{ for all } a\in\F_2^m.$$
		\end{enumerate}
	\end{theorem}

From this statement, it follows immediately that bent functions on $\F_2^m$ exist if and only if $m$ is even. Using the Walsh coefficients of a given bent function, one can always construct another Boolean function, known as its dual. 

\begin{definition}
	Let $f\in\boole(m)$ be bent. Then, the function $\tilde{f}\in\boole(m)$ defined by $(-1)^{\tilde f(a)}= 2^{-m/2}\walsh f(a)$ for all $a\in\F_2^m$ is called the \emph{dual} of $f$. 
\end{definition}

It is well-known that the dual of a bent function is bent~\cite{ROTHAUS1976}. Moreover, by Poisson's formula if $f$ is bent and abnormal  then its dual bent function $\tilde f$ is abnormal too~\cite{CHARPIN}. Now, we present an algebraic decomposition of a given bent function into two near-bent functions with special spectral properties. But first, we give a definition of a near-bent function. 

\begin{definition}
	For even $m$, a Boolean function $g\in\boole(m-1)$ is called \emph{near-bent} if the Walsh spectrum of $g$ is equal to $\{0,\pm 2^{m/2}\}$.
\end{definition}
	The following facts will be used later while constructing a bent function from suitably chosen near-bent functions.
\begin{remark}
	1. Note that for a Boolean function $f\in\boole(m)$, the set $\{ |\walsh  f (a) | \colon a\in\F_2^m \}|$ is invariant under EA-equivalence. Indeed, applying an affine permutation to the input results in a permutation of the Walsh coefficients, while adding an affine function can only alter their signs. As a result, bentness and near-bentness are preserved under EA-equivalence.
	
	\noindent 2. By Parseval's identity, every near-bent function $g\in\boole(m-1)$ has exactly $2^{m-2}$ Walsh coefficients that are equal to zero. Moreover, the algebraic degree of a near-bent function $g$ in $m-1$ variables is at most $m/2$. For a detailed discussion, we refer to~\cite[Section 6.2.2]{Carlet2021}. 
	
	\noindent 3. We observe that bent and near-bent functions have only a few special value distributions. Observe that for any Boolean function  $f\in\boole(k)$ it holds that $\walsh f(0)=2^k-2\operatorname{wt}(f)$, thus $\operatorname{wt}(f)=2^{k-1}-\frac{1}{2}\walsh f(0)$. If $k$ is even and $f$ is bent, then 
	\begin{equation}\label{eq: bent weight}
		\operatorname{wt}(f)=2^{k-1}\pm 2^{k/2-1},
	\end{equation}
	since $\walsh f(a)\in\{\pm 2^{k/2}\}$ holds for all $a\in\F_2^m$. If $k$ is odd and $f$ is near-bent, then $\operatorname{wt}(f)=2^{k-1}$ or 
	\begin{equation}\label{eq: near-bent weight}
		\operatorname{wt}(f)=2^{k-1}\pm 2^{(k+1)/2-1},
	\end{equation}
	since $\walsh f(a)\in\{0,\pm 2^{(k+1)/2}\}$ holds for all $a\in\F_2^m$ in this case. 
\end{remark}
		
	Any Boolean function $f \in \boole(m)$ can be written as a \emph{concatenation} of two Boolean functions $g,h\in\boole(m-1)$, which we denote by $f=(g||h)$ and formally define as follows: 
	\begin{equation}\label{eq: concat}
		\begin{split}
			f(x_1,\ldots,x_{m-1},x_m) :=& (g||h)(x_1,\ldots,x_{m-1},x_m)\\
			= &(x_m+1)g(x_1,\ldots,x_{m-1}) + x_m h(x_1,\ldots,x_{m-1}).
		\end{split}
	\end{equation}

	%For even $m$, a Boolean function in $\boole(m)$ having Walsh spectrum $\{\pm 2^{m/2}\}$ is called bent function.  \textbf{it should be iff theorem!}
	
	It is clear that the Walsh coefficient of $f=(g||h)\in\boole(m)$ at $(a,\alpha)\in\F_2^m =\F_2^{m-1}\times\F_2$ is given by

	\begin{equation}
		\label{WALSH}
		\walsh f(a,\alpha ) = \walsh g(a) + (-1)^\alpha\,\walsh h(a). 
	\end{equation}
	The following result gives a necessary and sufficient condition for the concatenation $f=(g||f)\in\boole(m)$ of two Boolean functions $(g||f)\in\boole(m-1)$ to be bent. 
	\begin{theorem}\cite{ZhengZhang2000}
		A Boolean function $f=(g||h)\in\boole(m)$ is  bent  if and only if the functions $g,h\in\boole(m-1)$ are near-bent and satisfy 
		\begin{equation}\label{eq: complementary near-bent}
			\forall a\in \F_2^{m-1},\quad \walsh g(a)=0 \Longleftrightarrow  \walsh h(a)\ne 0. 
		\end{equation}
	\end{theorem}  
	Note that near-bent functions $g$ and $h$ satisfying the condition given in Eq~\eqref{eq: complementary near-bent} are called \emph{complementary}. 
	More precisely, if  $f=(g||h)\in\boole(m)$ is bent, then from the definition of the dual and the form of the Walsh coefficient of $f=(g||h)$ given in Eq.~\eqref{WALSH}, one can deduce that:
	\begin{align}
		\hfill \text{if } \walsh g(a) = 2^m, & \quad \mbox{ then } \walsh h(a) = 0 \text{ and } \tilde{f}(a, \alpha) = 0; \label{eq: (i)} \\ 
		\hfill \text{if } \walsh g(a) = -2^m, & \quad \mbox{ then } \walsh h(a) = 0 \text{ and } \tilde{f}(a, \alpha) = 1;  \label{eq: (ii)} \\
		\hfill \text{if } \walsh g(a) = 0, & \quad \mbox{ then } \walsh h(a) = \pm 2^m \text{ and } \tilde{f}(a, 0) + \tilde{f}(a, 1) = 1. \label{eq: (iii)}
	\end{align}
	
%	\begin{enumerate}[label=(\roman*)]
%		\item if $\walsh g(a)=2^m$ then $\walsh h(a)=0$ and $\tilde f(a,\alpha)=0$;
%		\item if $\walsh g(a)=-2^m$ then $\walsh h(a)=0$ and $\tilde f(a,\alpha)=1$; 
%		\item if $\walsh g(a)=0$ then $\walsh h(a)=\pm 2^m$ and $\tilde f(a,0)+\tilde f(a,1)=1$.
%	\end{enumerate}
	%\begin{problem} 
	%\subsection{Sieving}
	
	%\subsection{Complexity reduction}
	%Here we show that it is in fact to work with the representatives of equivalence classes.
	
	In the case if $g,h\in\boole(m-1)$ are near-bent and the resulting concatenation $f=(g||h)\in\boole(m)$ is bent, we call the function $f$ a \emph{bent expansion} of $g$. 
%			\begin{remark}
%		By Lemma~\ref{MONO}, if $f=(g||h)\in\boole(8)$ is abnormal then $g,h\in\boole(7)$ are abnormal too. Thus, in order to prove that all $8$-bit bent functions are normal or weakly normal, it is sufficient to prove that all bent expansions of abnormal $7$-bit near-bent functions are normal or weakly normal. Moreover, it is enough to test the expansions of representatives of EA-equivalence classes of abnormal near-bent functions $g\in\boole(2,4,7)$. In fact, this argument is valid in the general setting $g\in\boole(2,m/2,m-1)$. To see that, assume that a near-bent function $g\in\boole(m-1)$ is bent-expandable, i.e., there exists a near-bent function $h\in\boole(m-1)$ such that $f=(g||h)$ is bent on $\F_2^m$. Let $g'\in\boole(m-1)$ be a near-bent function EA-equivalent to $g\in\boole(m-1)$, i.e., there exists an affine permutation $A$ of $\F_2^{m-1}$ and an affine function $a\in\boole(m-1)$ such that $g'(x)=g(A(x))+a(x)$ holds for all $x\in\F_2^{m-1}$. Define a near-bent function $h\in\boole(m-1)$ by $h'(x)=h(A(x))+a(x)$, for $x\in\F_2^{m-1}$. Then, for all $x\in\F_2^{m-1}$ and $x_m\in\F_2$ it holds that
%		$$ f'(x,x_m):=(g'||h')(x,x_m)=(g||h)(A(x),x_m)+a(x).$$
%		Consequently, bent-expandability is a property of the whole EA-equivalence class of near-bent functions.
%	\end{remark}

	\begin{remark}\label{rem: abnormal near-bent}
		1. Let $f=(g||h)\in\boole(m)$, $m$ even, be abnormal. Assume that one of $g$ and $h$ in $\boole(m-1)$ is not abnormal (w.l.o.g, we consider $g$), i.e., there exist $\lceil (m-1)/2 \rceil$-flat $a+U$ of $\F_2^{m-1}$ s.t. $g|_{a+U}$ is constant or linear. Define an $m/2$-flat of $\F_2^{m}$ by $a'+U'=(a,0)+(U\times\{0\})$. From Eq.~\eqref{eq: concat}, it holds that $f|_{a'+U'}=g|_U$ is constant or linear, which is a contradiction to the assumption that $f$ is abnormal. Consequently, both functions $g$ and $h$ are abnormal.

		\noindent 2. Using the argument above, to prove that all $m$-bit bent functions are normal or weakly normal, it is sufficient to show that all bent expansions of abnormal $(m-1)$-bit near-bent functions are normal or weakly normal. Moreover, it is enough to test the expansions of representatives of EA-equivalence classes of abnormal near-bent functions $g\in\boole(2,m/2,m-1)$. To see that, assume that a near-bent function $g\in\boole(m-1)$ is bent-expandable, i.e., there exists a near-bent function $h\in\boole(m-1)$ such that $f=(g||h)$ is bent on $\F_2^m$. Let $g'\in\boole(m-1)$ be a near-bent function EA-equivalent to $g\in\boole(m-1)$, i.e., there exists an affine permutation $A$ of $\F_2^{m-1}$ and an affine function $a\in\boole(m-1)$ such that $g'(x)=g(A(x))+a(x)$ holds for all $x\in\F_2^{m-1}$. Define a near-bent function $h\in\boole(m-1)$ by $h'(x)=h(A(x))+a(x)$, for $x\in\F_2^{m-1}$. Then, for all $x\in\F_2^{m-1}$ and $x_m\in\F_2$ it holds that
		$$ f'(x,x_m):=(g'||h')(x,x_m)=(g||h)(A(x),x_m)+a(x).$$
		Consequently, bent-expandability is a property of the whole EA-equivalence class of near-bent functions.
	\end{remark}
	
	In the context of this study, we consider the case $m = 8$ variables.  Even in this setting, for a fixed near-bent function $g \in \boole(7)$ (considered as a representative of an EA-equivalence class), there are still too many near-bent functions $h \in \boole(7)$ to check whether they extend $g$ into a bent function $f = (g \| h)$.  In the following subsection, we propose an efficient algorithm to construct all possible bent expansions $f = (g || h)$ for a given near-bent function $g$.

	\subsection{Bent expansion of a near-bent function}	
	For a given near-bent function $g$ in $\boole(m-1)$, the problem is to determine all $h$ in $\boole(m-1)$ such that $f = (g || h) \in \boole(m)$ is bent.  To achieve this, we use the approach from \cite{MENG}, which consists of determining the dual of the expansion. Denoting by $\Phi\in\boole(m)$ the dual of the bent expansion $f = (g || h)\in\boole(m)$, we express $\Phi$ as the concatenation  
	$\Phi := (\Phi_0 || \Phi_1)$,  where the near-bent Boolean functions $\Phi_0, \Phi_1 \in \boole(m-1)$ must satisfy the following conditions for every $a\in\F_2^{m-1}$ according to Eqs.~\eqref{eq: (i)}, ~\eqref{eq: (ii)} and~\eqref{eq: (iii)}. Accordingly, these functions are defined in the following way:
	\begin{equation}\label{eq: Phi0Phi1}
		\Phi_0(a) =\begin{cases} 0, & \walsh g(a)=+2^{m/2}\\
			1, & \walsh g(a)=-2^{m/2}
		\end{cases}\quad\mbox{and}\quad
		\Phi_1(a)  =\begin{cases} \Phi_0( a ), & \walsh g(a)=\pm2^{m/2}\\
			1+\Phi_0(a), & \walsh g(a)=0
		\end{cases}.
	\end{equation}
	Note that the function $ \Phi_0 + \Phi_1 \in \boole(m-1) $ is the indicator function of the set $ { a \in \mathbb{F}_2^{m-1} \mid \walsh g(a) = 0 } $, and there are exactly $ 2^{m-2} $ Walsh coefficients of $ g \in \boole(m-1) $ that are equal to zero. In this way, to completely determine the functions $ \Phi_0 $ and $ \Phi_1 $ defined by Eq.~\eqref{eq: Phi0Phi1}, one has to correctly assign values to the $ 2^{m-2} $ undefined entries of these functions. For large values of $m$, this becomes a complex task. In the following example, we illustrate the key ideas behind a more efficient method for determining these values, and hence bent-expansions.
	\begin{example}\label{ex: bent-expansion}		
		Let $m=6$. Consider the near-bent function $g\in\boole(m-1)$ defined in Example~\ref{ex: 1} as follows:
		$$g(x)=x_1 x_4 + x_2 x_4 + x_3 x_4 + x_2 x_3 x_4 + x_2 x_5 + x_3 x_5 + x_1 x_3 x_5.$$
		After setting $2^{m-2}=16$ values of the functions $\Phi_0$ and $\Phi_1$ partially defined by Eq.~\eqref{eq: Phi0Phi1}, there are still $2^{m-2}$ undefined values $u_i\in\F_2$ (and their complements $\bar u_i=u_i+1$) that  need to be filled, as illustrated by Figure~\ref{EXAMPLE}. 
		
		\setlength{\tabcolsep}{1.5pt}
		\vspace{-0.5cm}
		\begin{figure}[h]
			\caption{How to expand a near-bent function?}
			\footnotesize
			%\tiny
			\centering
			\begin{tabular}%{|*{33}{m{10pt}|}}
				{|c|c|c|c|c|c|c|c|c|c|c|c|c|c|c|c|c|}
				%{|X |X|X|X|X|X|X|X|X|X|X|X|X|X|X|X|X|X|X|X|X|X|X|X|X|X|X|X|X|X|X|X|X|}
				\hline
				\cellcolor[gray]{.8}$\hat g$& $+8$  & $0$  & $+8$  & $0$  & $0$  & $+8$  & $0$  & $+8$  & $+8$  & $0$  & $0$  & $-8$ & $0$  & $-8$ & $+8$  & $0$   \\ \hline
				\cellcolor[gray]{.8} $\Phi_0$& \PLUS &  $u_1$ & \PLUS  & $u_2$  & $u_3$  & \PLUS  &  $u_4$ & \PLUS  & \PLUS & $u_5$ & $u_6$  & \MOINS &  $u_7$ & \MOINS & \PLUS  & $u_8$  \\ \hline
				\cellcolor[gray]{.8} $\Phi_1$& \PLUS &  $\bar u_1$ & \PLUS  & $\bar u_2$  & $\bar u_3$  & \PLUS  &  $\bar u_4$ & \PLUS  & \PLUS & $\bar u_5$ & $\bar u_6$  & \MOINS &  $\bar u_7$ & \MOINS & \PLUS  & $\bar u_8$ \\ \hline \hline
				\cellcolor[gray]{.8}$\hat g$& $+8$  & $-8$ & $-8$ & $+8$  & $0$  & $0$  & $0$  & $0$  & $+8$  & $+8$  & $0$  & $0$  & $0$  & $0$  & $-8$ & $-8$ \\ \hline
				\cellcolor[gray]{.8} $\Phi_0$& \PLUS  & \MOINS & \MOINS & \PLUS  & $u_9$ &  $u_{10}$ &  $u_{11}$ & $u_{12}$ &  \PLUS & \PLUS  & $u_{13}$ & $u_{14}$ & $u_{15}$  & $u_{16}$  & \MOINS & \MOINS \\ \hline
				\cellcolor[gray]{.8} $\Phi_1$& \PLUS  & \MOINS & \MOINS & \PLUS  & $\bar u_9$ &  $\bar u_{10}$ &  $\bar u_{11}$ & $\bar u_{12}$ &  \PLUS & \PLUS  & $\bar u_{13}$ & $\bar u_{14}$ & $\bar u_{15}$  & $\bar u_{16}$  & \MOINS & \MOINS \\ \hline
			\end{tabular}
			\label{EXAMPLE}
		\end{figure}
    First, we observe that many choices of $u_i$ can be excluded immediately since the near-bent functions $\Phi_0$ and $\Phi_1$ can only have three possible Hamming weights, as given in Eq.~\eqref{eq: near-bent weight}. Based on these weight constraints, many distributions of $u_i$ are already ruled out because they produce incorrect Hamming weights. However, even when the correct weight is obtained, it does not necessarily guarantee a bent expansion, as the conditions on the Walsh spectrum must also be satisfied.
	\end{example}
    Now, we examine in detail the possible Walsh spectra of restrictions of bent functions.
	%To summarize the example above, there are two important things to keep in mind in the bent expansion process: expansions need to have correct weight, be abnormal (according to a remark above). However, the most important is the fact that the near-bent functions must also have correct spectra. 
	%\phil{I rewrite all from here to Lemma 5.12! I change the role of $r$, every where $r$ appears must be carrefully checked !}
	For a Boolean function $f\in\boole(m)$, 
	we define  $f_\lambda \in\boole( m - r )$ as the   restriction of $f$ to 
	the affine space $\fdmr\times\{\lambda\}$, i.e.,
	$$f_\lambda \colon  \fd^{m-r} \rightarrow \F_2, \quad t\in\fd^{m-r} \mapsto f_\lambda(t) = f(t,\lambda),\quad\mbox{for }\lambda\in\F_2^r,$$
	and call it \emph{$\lambda$-restriction}. In the following statement, we compute the Walsh spectrum of the functions $f_\lambda$ that originate from bent functions $f$; an alternative proof of this statement can be found in~\cite{MENG2008_DM}. 
	\begin{lemma} If $f\in\boole(m)$ is a bent function, 
		the Walsh spectrum of $f_\lambda\in\boole( m -r )$ is included in the set
		\begin{equation}\label{eq: Wr}
			W_r := \{  2^{m/2} -  i\times 2^{m/2 - r +1} \mid  0\leq i \leq  2^{r} \}.
		\end{equation}
	\end{lemma}
	\begin{proof}
		Observe that the orthogonal complement of $V=\fdmr\times\{0\}$ is the vector space $V^\perp = \{0\} \times  {\F^{r}_2}$ . Applying Poisson's formula, we get the following relation between the Walsh coefficients of the function $f_\lambda$ and the original function $f$:
		\begin{align*}
			\sum_{ v\in \fdmr} (-1)^{ f_\lambda( v ) + t\cdot v } & =\sum_{ v \in  \fdmr } (-1)^{ f( (v,0) + (0,\lambda)   ) + (t,0)\cdot (v,0)  } \\
			&= \frac 1{2^{r}} \sum_{w \in \F^{r}_2} \walsh f( t,w ) (-1)^{w.\lambda}.
		\end{align*}
        Since all the Walsh coefficients of $f$ are equal to $\pm 2^{m/2}$, the right-hand side of the latter equation is equal to $2^{m/2} \times 2^{ -r } \times  ( 2^{r} - 2 \times i  )$, for some integer $i$ satisfying $0 \le i \le 2^{r}$.
	\end{proof}
	
    Using this result, the authors of \cite{MENG} used an elegant approach to constructing the dual $\Phi$ of a bent function $f=(g||h)$ step by step. They achieved this by concatenating appropriate restrictions $\Phi_\lambda$ at different levels. This method effectively eliminates restrictions with incorrect Walsh coefficients, thereby reducing computational complexity. We illustrate this approach in the following figure for the case where the bit-word $\lambda$ has length $r\in \{1,2,3\}$.
	\vspace{-0.5cm}
	\begin{figure}[H]
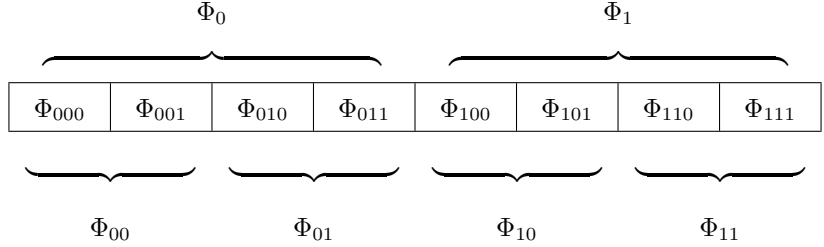

		\caption{Constructing the dual piece-by-piece: an illustration}
		\renewcommand{\arraystretch}{1.5} % Réduit l'espacement vertical entre les lignes
		\newcolumntype{C}[1]{>{\centering\let\newline\\\arraybackslash\hspace{0pt}}m{#1}}
		\begin{center}
			\begin{tabular}{|C{32pt}|C{32pt}|C{32pt}|C{32pt}|C{32pt}|C{32pt}|C{32pt}|C{32pt}|}
				\multicolumn{4}{c}{$\Phi_{0}$}
				&\multicolumn{4}{c}{$\Phi_{1}$}\\
				\multicolumn{4}{c}{  $\overbrace{\phantom{\hbox to 128pt{}}}^{}$ }
				&\multicolumn{4}{c}{  $\overbrace{\phantom{\hbox to 128pt{}}}^{}$ } \\
				\hline
				$\Phi_{000}$ & $\Phi_{001}$ & $\Phi_{010}$ & $\Phi_{011}$
				& $\Phi_{100}$ & $\Phi_{101}$ & $\Phi_{110}$ & $\Phi_{111}$
				\\
				\hline
				\multicolumn{2}{c}{  $\underbrace{\phantom{\hbox to 64pt{}}}_{}$ }
				&\multicolumn{2}{c}{  $\underbrace{\phantom{\hbox to 64pt{}}}_{}$ }
				&\multicolumn{2}{c}{  $\underbrace{\phantom{\hbox to 64pt{}}}_{}$ }
				&\multicolumn{2}{c}{  $\underbrace{\phantom{\hbox to 64pt{}}}_{}$ }
				\\
				\multicolumn{2}{c}{$\Phi_{00}$}
				&\multicolumn{2}{c}{$\Phi_{01}$}
				&\multicolumn{2}{c}{$\Phi_{10}$}
				&\multicolumn{2}{c}{$\Phi_{11}$} \\
			\end{tabular}
		\end{center}
		\label{fig: piece-by-piece}
	\end{figure}
	
	\def\admis#1;{{A}_{#1}}
	
	\begin{remark}
	In the remaining part of the paper, we consider \( r \in \{1,2,3\} \), 
	as it suffices to reduce the analysis of normality of bent functions in $m=8$ variables to Boolean functions in 7, 6, and 5 variables, respectively. We note that for $m\ge8$ one can consider higher values of $r$ without loss of generality.
	\end{remark}
	
	Given a bit-word $\lambda$ of length $r$, we denote $\lambda'$ the left prefix of $\lambda$ i.e.  $\lambda'b =\lambda$ for some bit $b$. With this notation, we introduce the following definition.

		\begin{definition}
		Let $m$ even, and let $g\in\boole(m - 1 )$ be a near-bent function. For a bit string $\lambda$ of length $r \in \{1,2,3\}$, we define the set of \emph{candidate functions} $C_\lambda(g)$, where each member $\phi \in \boole( m-r )$ 
		satisfies the following conditions:
		\begin{enumerate}
			\item $\forall t\in\fdmr,\quad \walsh g( t, \lambda') =+2^{m/2} \Rightarrow  \phi(t) = 0$; 
			\item $\forall t\in\fdmr,\quad \walsh g( t, \lambda') =-2^{m/2} \Rightarrow \phi(t) = 1$; 
			\item  The Walsh spectrum of $\phi$ is included  in the set $W_{ r  }$, which is defined in Eq.~\eqref{eq: Wr}.
		\end{enumerate}
		Members of the set $C_\lambda(g)$ are referred to as \emph{$\lambda$-candidates for $g$}.
		\end{definition}
	
	\begin{remark}
	Note that any $\lambda$-restriction of a bent expansion of $g$ is in  $C_\lambda(g)$
	but the converse is generally false. 
	\end{remark}
	
	All bent-expansions of $g\in\boole(m-1)$ can be derived from $C_0(g)$, a set obtained in \cite{MENG} through a bottom-up construction. This process involves merging $C_{0j}(g)$ by concatenating elements $\phi_L \in C_{0j0}(g)$ and $\phi_R \in  C_{0j1}(g)$, forming the functions $(\phi_L || \phi_R)$, while retaining only those that qualify as $0j$-candidates.  However, this method is inefficient, as for certain functions, the size of the corresponding $C_{ij}(g)$-sets can become excessively large. 
	
	In the following, we introduce our improvements that optimize this procedure by reducing the search space and improving computational efficiency.
	
	\begin{lemma}
		Let $\zeta\in\boole(m-1)$ be the indicator of Walsh zeroes of $g\in\boole(m-1)$. Let $\lambda$ be a bit string of length $r \in \{1,2,3\}$. If $\phi\in\boole( m - r )$ is a $\lambda$-restriction of a bent expansion of $g$, then the Boolean
		function defined by
		\begin{equation}
			v\mapsto \phi'( v ) = \phi(v) + \zeta(\lambda, v)
		\end{equation}
		is also $\lambda$-candidate for $g$.
	\end{lemma}
	\begin{proof}
		It is a consequence of Eq.~\eqref{eq: (iii)}.
	\end{proof}
	For an $r$-bit word $\lambda$, we denote by $A_\lambda(g)$ the set of Boolean functions $\phi \in \boole(m - r)$ such that both $\phi$ and $\phi'$ are $\lambda$-candidates. Its elements are called \emph{$\lambda$-admissible functions} for $g$. The cardinality of the set $A_\lambda(g)$ set may be very large, and we only have a trivial estimate:
	\begin{equation}
		\# A_\lambda( g ) \leq 2^{m_\lambda},
	\end{equation} 
	where $m_\lambda$ denotes the number of Walsh coefficients
	equal to 0 in the affine space $ \fdmr\times \{\lambda\}$.
	
	The following lemma will be used to minimize the size of the dictionary $\mathcal{D}$ in Algorithm \ref{EXPAND}, which serves as a lookup table that organizes and stores admissible functions $L$ based on their corresponding keys $\kappa$ (that we define in a moment). Its primary role is to facilitate efficient retrieval of functions $L$ that share the same key during the later stages of this algorithm.
	\begin{lemma}%[twist]
		Let $m$ be even, and let $g \in \boole(m-1)$ be near-bent. Without loss of generality, we may assume that
		$$
		\# A_{00}( g ) \leq 2^{m-3}. 
		$$
	\end{lemma}
	\begin{proof}
		Since $m_{00} + m_{01} = 2^{m-2}$, we know that one of these numbers is at most $2^{m-3}$. If it is $m_{00}$, there is nothing to do. Otherwise, we can translate $g$ by $x_{m-1}$, and the result follows.
	\end{proof}
	
	\begin{remark}
	Note that for the case we are interested in, i.e., $m = 8$, we have $\#A_{00}(g) \leq 2^{32}$.  
	\end{remark}

    For an $(m-2)$-bit function $\phi$, where $m$ is even, we define the \emph{key} $\kappa(\phi)$ of $\phi$ as the position $a \in \mathbb{F}_2^{m-2}$ such that $\walsh \phi(a) \equiv 2^{m/2 - 1} \mod 2^{m/2}$. The following lemma enables us to propose an efficient implementation of the ideas in \cite{MENG}.
	\begin{lemma}%[key]
		Let $m$ be even. If $\phi\in\boole(m-1)$ is 0-admissible for the near-bent function $g\in\boole(m-1)$, then $\phi$ is the concatenation $(\phi_L|| \phi_R)$, where $\phi_L$ is $00$-admissible, $\phi_R$ is $01$-admissible, and
		$$
		\kappa( \phi_L ) =  \kappa( \phi_R ).
		$$
	\end{lemma}
	\begin{proof}
		We know that $\phi_L$ and $\phi_R$ take values in $W_{m-2} = \{0,\pm 2^{m/2 - 1}, \pm 2^{m/2} \}$,
		and 
		$$\walsh \phi ( a ) = \walsh \phi_L  (a)  \pm \walsh  \phi_R (a) \in W_{m-1} = \{0,\pm 2^{m/2} \}.$$
		This completes the proof.
	\end{proof}

    \subsection{Main result}
	Using the ideas above, we propose the following algorithm, which can be used not only to check whether a given near-bent function is bent-expandable but also to find its bent expansions directly. In the following, we analyze the bent expansions of abnormal near-bent functions $g \in \boole(m-1)$ for $m = 8$, which were constructed in the previous section.
	\begin{algorithm}[H]
		\caption{$\text{EXPANSION}(g)$}
		\begin{algorithmic}[1]
			\Require A near-bent function $g$ in $\boole(m-1)$. %{\color{blue} Sasha: Looks like everything below is independent on $g$, what shouldn't be the case}
			%{\color{magenta} admissibily depend on $g$}
			\Ensure Print all bent expansions of $g$.
			\State Initialize a dictionary   $\mathcal{D}$.
			\State $A_{00}(g)$:= the set of $00$-admissible functions for $g$
			\For{all $L$ in $A_{00}(g)$}
			\State $\kappa$:= key($L$)
			\State add entry ($\kappa$,$L$) in $\mathcal{D}$
			\EndFor
			\State $A_{010}(g)$:= the set of $010$-admissible functions for $g$
			\State $A_{011}(g)$:= the set of $011$-admissible functions for $g$
			\For{all $(\phi_{010}, \phi_{011})$ in  $A_{010}(g) \times A_{011}(g)$  }
			\State $R$ := $(\phi_{010}||\phi_{011} )$
			\If  {$R$ is 01-admissible for $g$}
			\State $\kappa$ := key($R$)
			\For{all $L$ with key  $\kappa$} 
			\If{$(L||R)$ is 0-admissible for $g$}
			%\State $\Phi$ := $( L || R || L' ||  R' )$
			\State $\Phi_0$ := $( L || R )$ 
			\State $\Phi_1$ := $(  L' ||  R' )$ 
			\If{$(  \Phi_0 ||  \Phi_1 )$ is bent}
			\State print its dual!
			\EndIf
			\EndIf
			\EndFor
			\EndIf
			\EndFor
		\end{algorithmic}
		\label{EXPAND}
	\end{algorithm}

    The set of all abnormal functions obtained using the sieving procedure described in Algorithm~\ref{SIEVING} contains $4\,987$ near-bent functions. We identify $1\,309$ representatives up to extended-affine equivalence using Magma with the method described in Section~\ref{sec 2: intro}. Consequently, we apply Algorithm~\ref{EXPAND} to determine which of these near-bent functions are bent-expandable and, in turn, to identify all bent expansions. It appears that only 30 of these representatives are bent-expandable. These are listed in Appendix~\ref{appendix}.
	
    At the end, we collect $\approx 2^{24}$ bent functions (all of which are available on the web page of the project~\cite{project_normal}). Applying Algorithm~\ref{ABNORMAL}, we check that all of them are normal or weakly normal. Thus, we obtain the main result of this paper:
	
    \begin{theorem}
        All the bent functions of $\boole(8)$ are normal or weakly normal.
    \end{theorem}

%	{\color{blue} Sasha: Two comments about this discussion:
%		
%		\begin{enumerate}
%			\item My first problem is that now we have three different numbers: 39952, $12\times1+ 16\times1+24\times1+\cdots+1\times 1535+0\times66535=39888$ and now 43860. Should it be only one number everywhere? We also need to connect Sections 3 and 4 nicely: Surely, in Section 3 we investigate abnormality of quartics (so it is an interesting section on its own), however, since we use these abnormal functions in Section 4, it might be also considered as a warm-up section, preparing the necessary computations for Section 4. This would give a much better flow to the paper.
%			\item Another not a problem, but rather a suggestion. I think that reviewers of such papers like when the most essential information (like some of the data) is presented in the form of tables. Surely, we can't list here all 43860 functions or even 1309. What I think is possible to do is to list these most important 30 representatives (by ANF) and for each of them to give a number of bent extensions (similarly to as it is done in Table~\ref{Nh-values}). You can send me the data: list of representatives and list with corresponding multiplicities, and I will make the tables.
%		\end{enumerate}
%		
%	}
%	
%	{\color{magenta} Phil: two numbers :-)}
%	
%	{\color{magenta} Phil: I will recompute the 30's as soon as APNs will stop (saturday)}
	
	\section{Conclusion}\label{sec 6: conclusion} 
	
	In this work, we investigated the normality of Boolean functions, with a particular focus on bent functions in dimension $m= 8$. Our main contributions include the determination of the values $D^\dag_r(k,6)$ and $D^\dag_r(k,7)$, the classification of abnormal 7-bit quartic functions, and the proof that all 7-bit cubic functions are normal. Furthermore, we established that all 8-bit bent functions are either normal or weakly normal, thus determining their unifying algebraic property in this dimension. We believe that the ideas presented in this article can be applied to address the following questions, thereby advancing the understanding of the normality of Boolean (bent) functions.
	
	\begin{enumerate}
		\item The next open case concerns the normality of cubic bent functions in dimension $m=10$, where a few known examples outside the completed Maiorana–McFarland class have been shown to be normal (see \cite{Polujan2020} for details). 
		\item More generally, for even dimensions $m \geq 10$, it remains an open question whether all cubic bent functions in $B(m)$ are necessarily normal or weakly normal. Finding (if possible) cubic bent functions $f\in\boole(m)$ with $\deg_{ m/2 }(f)\ge 1$ is thus a very interesting and challenging research task. 
		%\item One possible approach to answering the previous question is by classifying cubic bent functions in ten variables, a problem that remains unsolved.
		%\item More generally, it is unknown whether cubic (not necessarily bent) Boolean functions in higher even dimensions are always normal or weakly normal, with the next open dimension $m = 10$ being of particular interest.  
		\item Finally, we conjecture that all 8-bit quartic functions are normal or weakly normal.
		%\item Characterize non-expandable near-bent functions (and give one example in the paper). If g and h complementary near-bent, then g+h is near bent too? Non-extendable near-bents are hence never parts of AB functions?? Non-extendable bent functions were studied in...
	\end{enumerate} 
	
	\section*{Acknowledgments}
	Philippe Langevin is partially supported by the French Agence Nationale de la Recherche through the SWAP project under Contract ANR-21-CE39-0012.

	%\bibliographystyle{IEEEtranS}
	%\bibliography{bibiliography_submission.bib}
	
    	% Generated by IEEEtranS.bst, version: 1.14 (2015/08/26)

	%\newpage
	\appendix
%\section{Appendix: Supplementary computations}
%\subsection{Functions with optimal relative degree}
%In this section, we give the algebraic normal forms of Boolean functions providing the lower bounds for the relative degree in Table~\ref{DR7}.
%\begin{itemize}
%	\item $f\in B(?,4,7)$ s.t. $\deg_6(f)=4$, thus giving $D^\dag_{6}(4,7)\ge 4$
%	$$f(x)= $$
%	\item $f\in B(?,5,7)$ s.t. $\deg_6(f)=4$, thus giving $D^\dag_{6}(5,7)\ge 4$
%	$$f(x)= $$
%	\item $f\in B(?,7,7)$ s.t. $\deg_6(f)=4$, thus giving $D^\dag_{6}(7,7)\ge 4$
%	$$f(x)= $$
%	\item $f\in B(?,4,7)$ s.t. $\deg_5(f)=3$, thus giving $D^\dag_{5}(4,7)\ge 3$
%	$$f(x)= $$
%	\item $f\in B(?,5,7)$ s.t. $\deg_5(f)=3$, thus giving $D^\dag_{5}(5,7)\ge 3$
%	$$f(x)= $$
%	\item $f\in B(?,6,7)$ s.t. $\deg_5(f)=3$, thus giving $D^\dag_{5}(6,7)\ge 3$
%	$$f(x)= $$
%	\item $f\in B(?,7,7)$ s.t. $\deg_5(f)=3$, thus giving $D^\dag_{5}(7,7)\ge 3$
%	$$f(x)= $$
%	\item Also give the functions for the bounds, $1\leq D^\dag_4( k, 7)$, provided we don't have a theoretical explanation for them.
%\end{itemize}
%
%\subsection{Abnormal bent-expandable near-bent functions}
\newpage
\section{Abnormal bent-expandable near-bent functions}\label{appendix}
In this section, we list all 30 EA-inequivalent abnormal bent-expandable near-bent functions in dimension 7. Every near-bent function is given by a string, representing the truth table in hexadecimal format.

%\numberwithin{myListing}{section}
%\begin{myListing}
%    \caption{Abnormal bent-expandable near-bent functions with Sage}
    \begin{center}
        \includegraphics[scale=0.95]{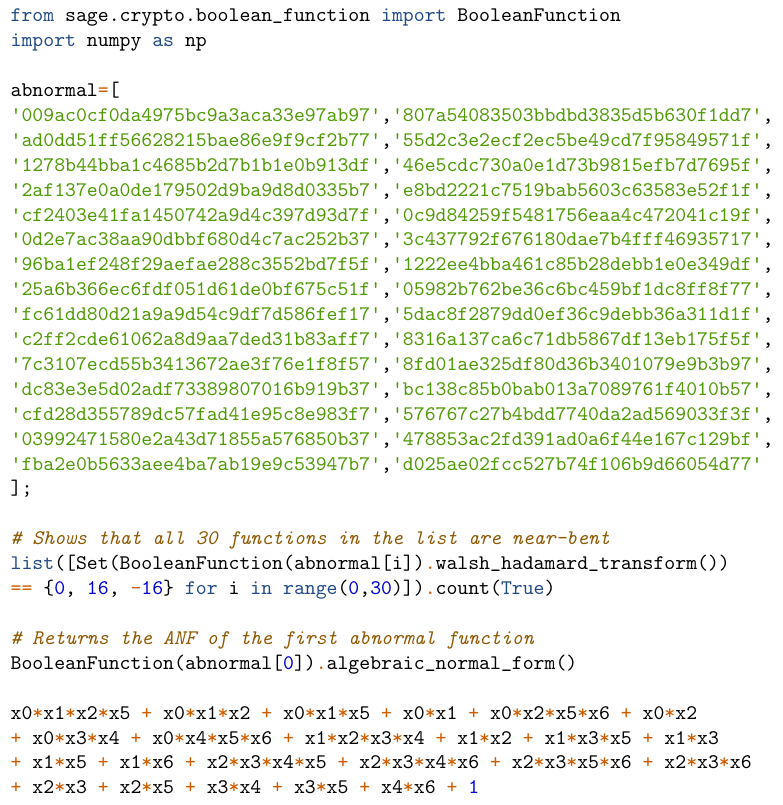}
    \end{center}

\end{document}